%% file: main-arXive.tex
\documentclass{article}
\usepackage{fullpage,subfigure}
\usepackage{multirow}
\usepackage{amsmath,amssymb, latexsym,graphicx}
\usepackage{url}
\usepackage{graphicx}
\usepackage[linesnumbered,lined,ruled,noend]{algorithm2e}
\usepackage{array}
\usepackage{color}
\usepackage{stmaryrd}

\newcommand{\sizep}{\mathrm{size}}
\newcommand{\comp}{\mathrm{comp}}

\newcommand{\removed}[1]{}

\newcommand{\tobinary}{\textsc{ToBinary}}

\newcommand{\Primes}{\mathcal{P}}

\newcommand{\iz}{indicator}
\newcommand{\ind}{indicator}
\newcommand{\each}{\emph{\bf each} }

\newcommand{\rr}{p}

\newcommand{\enc}[1]{\llbracket #1 \rrbracket}
\newcommand{\Decodesketch}{\textsc{Decode}}

\newcommand{\ZZ}{{\mathbb Z}}

\newcommand{\spirit}{\mathrm{SPiRiT}}

\newcommand{\cost}{array}

\newcommand{\br}[1]{\left\{#1\right\}}

\newcommand{\danny}[1]{}
\newcommand{\anote}[1]{\textcolor{red} {}}
\newcommand{\remove}[1]{}

\newcommand{\f}{\frac}

\newcommand{\sq}{\hbox{\rlap{$\sqcap$}$\sqcup$}}
\newcommand{\qed}{\hspace*{\fill}\sq}
\newcommand{\eps}{\ensuremath{\varepsilon}}                       
\newcommand{\REAL}{\ensuremath{\mathbb{R}}}                       
\newcommand{\RR}{{\REAL}}

\newenvironment{proof}{\noindent {\bf Proof}.\ }{\qed \par\vskip 4mm\par}

\newtheorem{theorem}{Theorem}[section]

\newtheorem{lemma}[theorem]{Lemma}

\newtheorem{definition}[theorem]{Definition}

\newtheorem{remark}[theorem]{Remark}
\newtheorem{claim}[theorem]{Claim}

\newcommand{\size}{s}

\newcommand{\oneWitness}{\textsc{SearchCoresetItem}}

\newcommand{\isZero}{is\textsc{Zero}}
\newcommand{\isEq}{is\textsc{Equal}}

\newcommand{\isMatch}{{is\textsc{Match}}}

\newcommand{\isPositive}{\textsc{isPositive}}

\newcommand{\sfirstindex}{\textsc{SPiRiT}}
\newcommand{\securefirstindex}{\textsc{SecureSearch}}
\newcommand{\securesearch}{\textsc{SecureSearch}}

\newcommand{\tree}{\mathcal{T}}
\def\parents{\textsf{Ancestors}}
\def\lop{\textsf{Siblings}}

\newcommand\sett[2]{\left\{ \left. #1 \;\right\vert #2 \right\}}
\newcommand\set[1]{{\left\{ #1 \right\}}}

\newcommand\card[1]{\left| #1 \right|}
\newcommand\details[1]{}
\def\zo{\set{0,1}}
\def\ie{i.e.}
\def\eg{e.g.}

\newcommand{\ceil}[1]{\lceil #1\rceil}

\newcommand{\emailAdi}{akavia@mta.ac.il}
\newcommand{\emailDanny}{dannyf.post@gmail.com}
\newcommand{\emailHayim}{hayim.shaul@gmail.com}
\newcommand{\instAdi}{Cybersecurity Research Center, Academic College of Tel-Aviv Jaffa}
\newcommand{\instDanny}{Robotics \& Big Data Lab, University of Haifa}
\newcommand{\instHayim}{Robotics \& Big Data Lab, University of Haifa}
\newcommand{\grantsAdi}{}
\newcommand{\grantsDanny}{}
\newcommand{\grantsHayim}{}

\title{Secure Search on the Cloud via Coresets and Sketches}

\author{Adi Akavia\thanks{\instAdi. Email: \textsf{\emailAdi}. \grantsAdi} \and
        Dan Feldman\thanks{\instDanny. Email: \textsf{\emailDanny}. \grantsDanny} \and
        Hayim Shaul\thanks{\instHayim. Email: \textsf{\emailHayim}. \grantsHayim}}

\date{}

\begin{document}
\maketitle

\begin{abstract}
\emph{Secure Search} is the problem of retrieving from a database table (or any unsorted array) the records matching specified attributes, as in SQL SELECT queries, but where the database and the query are encrypted. Secure search has been the leading example for practical applications of Fully Homomorphic Encryption (FHE) starting in Gentry's seminal work; however, to the best of our knowledge all state-of-the-art secure search algorithms to date are realized by a polynomial of degree $\Omega(m)$ for $m$ the number of records, which is typically too slow in practice even for moderate size $m$.

In this work we present the first algorithm for secure search that is realized by a polynomial of degree polynomial in $\log m$. We implemented our algorithm in an open source library based on HELib implementation for the Brakerski-Gentry-Vaikuntanthan's FHE scheme, and ran experiments on Amazon's EC2 cloud. Our experiments show that we can retrieve the first match in a database of millions of entries in less than an hour using a single machine; the time reduced almost linearly with the number of machines.

Our result utilizes a new paradigm of employing coresets and sketches, which are modern data summarization techniques common in computational geometry and machine learning, for efficiency enhancement for homomorphic encryption. As a central tool we design a novel sketch that returns the first positive entry in a (not necessarily sparse) array; this sketch may be of independent interest.

\end{abstract}

\newpage
\setcounter{page}{1}

\input{sections/introduction2.tex}

\input{sections/problem-statement-search-2}
\input{sections/spirit-construction}

\input{sections/experimental_results_spirit2.tex}

\input{sections/generic-search-short-version.tex}

\section{Conclusions and Followup Work}

In this work we show how to solve the secure search problem in overall time that is poly-logarithmic in the input array size. Our techniques can be extended to securely solving further problems in overall time in the input size, including securely returning \emph{all} matching array entries
~\cite{AFSreport}, and for secure optimization and learning~\cite{AFSmin} (for the former the time is polynomial also in the number of matching entries).

\section{Acknowledgment}
We thank Shai Halevi and Shafi Goldwasser for helpful discussions and
comments.

\bibliographystyle{alpha}
\bibliography{secure-search}

\appendix
\input{sections/prelim2.tex}
\input{sections/proofs-search.tex}

\end{document}

%% file: sections/introduction2.tex
%
\section{Introduction}\label{sec:intro}
%

Storage and computation are rapidly becoming a commodity with an increasing trend of organizations and individuals (client) to outsource storage and computation to large third-party systems often called ``the cloud'' (server). Usually this requires the client to reveal its private records to the server so that the server would be able to run the computations for the client. With e-mail, medical, financial and other personal information transferring to the cloud, it is paramount to guarantee {\em privacy} on top of data availability while keeping the correctness of the computations.


\paragraph{Fully Homomorphic Encryption (FHE)} \cite{GentryThesis09,GentrySTOC09} is an encryption scheme with the special property of enabling computing on the encrypted data, while simultaneously protecting its secrecy; see a survey in~\cite{HaleviShoup2014helib}.
%
Specifically, FHE allows computing any algorithm on encrypted input (ciphertexts), with no decryption or access to the secret key that would compromise secrecy, yet succeeding in returning the encryption of the desired outcome. Furthermore, the computation is \emph{non-interactive} and with \emph{low communication}: the client only sends (encrypted) input $x$ and a pointer to a function $f$ and receives the (encrypted) output $y=f(x)$, where the computation is done on the server's side, requiring no further interaction with the client.

\paragraph{The main challenge }for designing algorithms that run on data encrypted with FHE is to present their computation as a \emph{low degree polynomial} $f$, so that on inputs $x$ the algorithm's output is $f(x)$ (see examples in \cite{Naehrig11,Graepel12,lauter2014private,hamming15,Cheon2015,LuKS16,CryptoNet16}).
Otherwise, a naive conversion of an algorithm for its FHE version might yield highly impractical result.

\paragraph{Secure search} has been the hallmark example for useful FHE applications and the lead example in Gentry's PhD dissertation \cite{GentryThesis09}.
Here, the goal of the client is to search for an entry in an unsorted array, based on a lookup value.
%
In the secure version, 
the server gets access only to encrypted versions for both the array 
and the lookup value, and returns an encrypted version of the index of the desired entry; 
see formal Definition~\ref{def:search}. 

Use case examples include
secure search for a document matching a retrieval query in a corpus of sensitive documents, such as private emails, classified military documents, or sensitive corporate documents (here each array entry corresponds to a document specified by its list of words or more sophisticated indexing methods, and the lookup value is the retrieval query).
%
%
Another example is a secure SQL search query in a database, e.g., searching for a patient's record in a medical database based on desired attributes (here the array correspond to database columns, and the lookup value specifies the desired attributes).
%


Nevertheless, despite the centrality of the secure search application, little progress has been made in making secure search feasible in practice, except for the 
settings where: (i) either there are only very few records in the database so that the server can return the entire indicator vector indicating for each entry whether or not it is a match for the lookup value~\cite{Cheon2015,hamming15}
or (ii) the lookup value is guaranteed to have a \emph{unique match} in the array~\cite{DorozSunar2014pir,DorozSunar2016pir4search}, or (iii) for an unbounded number of matches the server returns all matches in time is polynomial in the number of matches \cite{AFSreport}.
Clearly, this is insufficient for the 
for use cases of secure search of large datasets and no a-priori unique of few matches guarantee.

\subsection{Our Contribution}

Our contributions in this work are as follow.

\paragraph{First 
solution to the secure search problem} that is applicable to large datasets with unrestricted number of matches. Furthermore, it is non-interactive and with low-communication to the client.
This is by suggesting a search algorithm that can be realized by a polynomial of degree that is only poly-logarithmic in the number $m$ of input records. Prior solutions with low degree polynomials are only for the restricted search settings discussed above; whereas for the unrestricted search problem, the folklore natural polynomial has linear degree in $m$. See Theorem~\ref{lem:FirstPositive-ClientSide} for details.

\paragraph{System and experimental results for secure search.} We implemented our algorithms into a system that runs on up to 100 machines on Amazon's EC2 cloud. Our experiments demonstrating, in support of our analysis, that on a single computer we can answer search queries on database with millions of entries in less than an hour. Furthermore performance scales almost linearly with the number of computers, so that we can answer search queries of a database of billions of entries in less than an hour, using a cluster of roughly 1000 computers; See Section~\ref{sec:xp} and Figure~\ref{fig:spirit_graphs}.

\paragraph{High accuracy formulas for estimating running time} that allow potential users and researchers to estimate the practical efficiency of our system for their own cloud and databases; See Section~\ref{sec:formula} and Figure~\ref{fig:spirit_graphs}.

\paragraph{Open Source Library for Secure-Search with FHE}    based on HELib~\cite{HELib} is provided for the community~\cite{open}, to reproduce our experiments, to extend our results for real-world applications, and for practitioners at industry or academy that wish to use these results for their future papers or products. 

\paragraph{Extensions.}
Our solution extends to retrieving not only exact matches to the lookup value, but also \emph{approximate matches} with respect, say, to a Hamming/edit/Euclidean distance bound, or any generic $\isMatch()$ algorithm; see Section~\ref{sec:generic-search}.
Our solution easily extends to returning not only the index $i^*$ of the first entry in $\cost$ matching the lookup value $\ell$, but also its content $\cost(i^*)$ (similarly, database row $i^*$), e.g. by utilizing techniques from~\cite{DorozSunar2014pir,DorozSunar2016pir4search}.


\subsection{Novel Technique: Search Coreset and SPiRiT Sketch}

\paragraph{Coreset }is a data summarization $C$ of a set $P$ of items (e.g. points, vectors or database records) with respect to a set $Q$ of queries (e.g. models, shapes, classifiers, points, lines) and a loss function $f$, such that $f(P,q)$ is approximately the same as $f(C,q)$ for every query $q\in Q$. The goal is to have provable bounds for (i) the size of $C$ (say, $1/\eps$), (ii) the approximation error (say, $\eps\in(0,1)$) and (iii) construction time of $C$ given $(P,Q,f)$. We can then run (possibly inefficient) existing algorithms and heuristics on the small coreset $C$, to obtain provably approximated solution for the optimal query (with respect to $f$) of the original data.

The coreset is a paradigm in the sense that its exact definition, structure and properties change from paper to paper. Many coresets were recently suggested to solve main problems e.g. in computational geometry (e.g.~\cite{phillips2016coresets,clarkson2010coresets,braverman2011streaming}), machine learning~(e.g.\cite{huggins2016coresets,braverman2017clustering}), numerical algebra~\cite{woodruff2014sketching} and graph theory (e.g.~\cite{czumaj}).

\paragraph{New Search Coreset for Homomorphic Encryption. }In this paper we suggest to solve the secure search problem using coreset. The idea is that instead of requiring that \emph{all} the computation will be done on the server, \emph{most} of the computation will be done on the server. More precisely, the server sends to the client only a small set of encrypted indices, that we call \emph{search coreset}. These coresets are small, 
and communicating them to the client, decrypting them, and decoding the desired result from the coreset is fast and require only little time on the client side. Concretely, in this work the coreset size and decoding time is polynomial in the output size, $\log m$, where $m$ is the input input array size. While it is not clear whether the search problem can be solved efficiently (i.e., via low-degree polynomials) only using the server, we prove in this paper that efficient and secure computation of our suggested search coreset on the server side is possible. Moreover, it reduces the existing state of the art of client computation time from exponential to polynomial in the output size.

Unlike traditional coreset, in this paper the coresets are exact in the sense that the data reduction does not introduce any additional error $\eps$. Moreover, the goal is not to solve an optimization problem, but to suggest algorithm that can be realized by a low-degree polynomial.

\paragraph{Sketches} can be considered as a special type of coresets, where given an $n\times d$ matrix $P$ and a $\log n\times n$ (fat) matrix $S$, the result $C=SP$ is a $\log n\times d$ (``sketch'') vector. Many problems can be efficiently solved on the short sketch $C$ instead of the long matrix $P$, by designing a corresponding (sketch) matrix $S$. See examples and surveys in~\cite{woodruff2014sketching,kane2010exact,clarkson2013low,li2014turnstile}. In this paper we focus on vectors $P$ (setting $d=1$).

\paragraph{SPiRiT Search Sketch} is one of our main technical contribution. It is a search sketch that gets a non-negative vector and returns the index of its first positive entry. It has independent interest beyond secure encryption, since unlike other Group Testing sketches (e.g.~\cite{indyk2010efficiently}) it can be applied on non-sparse input vectors.

To use it, the server first turns the input vector into an (encrypted) indicator binary vector whose $i$th entry is 1 if and only if the $i$th input entry matches the desired lookup value (where a ``match'' is either by an equality condition or a generic condition such as Hamming/edit/Euclidean distance bound; see Algorithms \ref{alg:ToBinaryExact} and Section \ref{sec:generic-search} respectively). Then the search sketch is applied on the indicator vector, and the output is sent to the client.

This $S\circ P\circ i\circ R\circ i\circ T$ sketch is a composition of a \textbf{S}ketch, 
\textbf{P}airwise, \textbf{R}oot, and \textbf{T}ree matrices, together with a binary operator $i(x)=\isPositive(x)$ that returns $1$ for each strictly positive entry in $x$, and $0$ for $x=0$. The design and proof of correctness of this new sketch for non-negative reals (independent of FHE) is described in Section~\ref{sec:spirit-construction}.

\paragraph{SPiRiT for FHE.} We observe that sketches have an important property that makes them very relevant to FHE: they can be implemented via low-degree (linear) polynomial. In fact, the matrices of our SPiRiT are binary and thus can be implemented as sums of subsets, without a single multiplication. Nevertheless, a naive realization would require to apply the polynomial over a large ring $p$, which (mainly due to the $\isPositive$ operator, see Lemma~\ref{lem:ispositive}) requires a polynomial of degree linear in the length $m$ of the input array.

To reduce the degree to poly-logarithmic in input array length $m$, in Sections~\ref{sec:SearchCoresetItem} we introduce a more involved version of SPiRiT that can be applied on small rings, but may not return the correct search result. The server executes this version over $k=(\log m)^{O(1)}$ rings and returns the resulting $k$-coreset to the client. Using techniques such as the Chinese Remainder Theorem and a proper selection of prime ring sizes, we prove that the client is guaranteed to decode the desired index efficiently from this coreset.

%
%

%% file: sections/problem-statement-search-2.tex
%
\section{Problem Statement and Main Result}
\label{sec:problem-statement}
%

In this section we give a formal statement of the search problem and our main result of efficient algorithm for secure search.

\subsection{The Secure Search Problem}

The goal of this paper is to solve the following search problem on encrypted input array and lookup value,
efficiently under parallel secure computation model formally defined in this section.

Denote $[m]=\br{1,\cdots,m}$; the entries of a vector $\cost$ of length $m\geq1$ by $\cost=(\cost(1),\cdots,\cost(m))$; and every vector is a column vector unless mentioned otherwise.
Denote by $\enc{x}$ the encrypted value of a vector or a matrix $x$ (where encryption of vectors is entry-by-entry and of values represented by $t$ digits, \eg, binary representation, is digit-by-digit).

\begin{definition}[Search Problem\label{def:search}.]
The goal of the \emph{search problem}, given a (not necessarily sorted) $\cost\in \br{0,\cdots,r-1}^m$ and a lookup value $\ell\in\br{0,\cdots,r-1}$, is to output
\[
i^* = \min\sett{i\in[m]}{\cost(i)=\ell}
\]
\end{definition}
Here, and in the rest of the paper we assume that the minimum of an empty set is $0$.

The \emph{Secure-Search} problem is the Search problem when computation is on encrypted data.
That is, the input is ciphertexts $\enc{\cost}$ and $\enc{\ell}$ encrypting $\cost$ and $\ell$ respectively, and the output is a ciphertext $\enc{f(\cost,\ell)}$ encrypting the desired outcome $f(\cost,\ell)$.
Our computation model only assumes that the encryption is by any fully (or leveled) homomorphic encryption (FHE) scheme, e.g.~\cite{BGV12},
%
Such encryption schemes satisfy that given ciphertexts $c_1=\enc{x_1},\ldots,c_n=\enc{x_n}$ one can evaluate polynomials $f(x_1,\ldots,x_n)$ on the plaintext data $x_1,\ldots,x_n$ by manipulating only the ciphertexts $c_1,\ldots,c_n$, and obtaining as the outcome a ciphertext $c = \enc{f(x_1,\ldots,x_n)}$; see \cite{haleviSurvey17} for a survey of FHE, and open implementations of such encryptions e.g. in HELib~\cite{HELib}.
Here and throughout the paper $f()$ is a polynomial over the finite ring $\ZZ_p$ of integers modulo $p$ (see Definitions~\ref{polydef},\ref{ring}), where $p$ is a parameter chosen during encryption, called the \emph{plaintext modulo}; $f(\cost,\ell)$ is the outcome of evaluating $f()$ when assigning values $\cost,\ell$ to its undetermined variables.

In the context of our coreset paradigm, the output $y = f(\cost,\ell)$ is a short sketch, named, \emph{Search Coreset}, so that there is an efficient decoding algorithm to obtain from $y$ the smallest index $i^*$ where $\cost(i^*)=\ell$.

\begin{definition}[Search Coreset] \label{def:sketch-firstpositive}
Let $k,m\ge 1$ be a pair of integers, $\cost\in \REAL^m$ and $\ell\in\REAL$. A vector $y\in \REAL^k$ is a \emph{$k$-search coreset} for $(\cost,\ell)$ if, given only $y$, we can decode (compute) the smallest index $i^*$ of $\cost$ that contains the lookup value $\ell$, or $i^*=0$ if there is no such index.
\end{definition}

The usage scenario is that the server computes $\enc{y} = \enc{f(\cost,\ell)}$ while seeing encrypted values only, whereas the client decrypts and decodes to obtain the desired outcome $i^*$:

\begin{definition}[The Secure-Search Problem]\label{def:secure-search}
Let $k$, $m,r$, $\cost$, $\ell$ and $i^*$ be as in Definition~\ref{def:search}.
The Search problem on $(\cost,\ell)$ is \emph{securely solved} by a non-interactive protocol between a server and a client with shared memory holding ciphertext $\enc{\cost}$ with plaintext modulo $p$
if:
\begin{enumerate}
    \item
        The client sends to the server encrypted lookup value $\enc{\ell}$ and the corresponding ring modulus $p$.
    \item
    \label{enum:call-to-server}
        The server evaluates a polynomial $f(\cost,\ell)$ over $\ZZ_p$ using homomorphic operations to obtain a ciphertext $\enc{y}=\enc{f(\cost,\ell)}$ of a $k$-search-coreset $y$ for $(\cost,\ell)$, and sends $\enc{y}$ to the client.
    \item
        The client decrypts $\enc{y}$ and decodes $y$ to obtain the smallest index  $i^*$ where $\cost(i^*)=\ell$.
\end{enumerate}
The \emph{server time} is $O(d+\log s)$ for $d,s$ the degree and size of the polynomial $f()$; see Definition~\ref{polydef}.
The \emph{client time} is the time to decode $y$. The \emph{overall time} is the sum of client time and server time.

More generally, the server may compute several polynomials $f_1(),f_2(),\ldots$. Moreover, the polynomials may be computed over distinct plaintext moduli $p_1,p_2,\ldots$, provided that the server has ciphertexts $\enc{\cost}_{p_j}$ and $\enc{\ell}_{p_j}$ corresponding for each plaintext moduli $p_j$. The \emph{server time} in this case is $O(d+\log s)$ for $d,s$ the maximum degree and size over all polynomials.
We call a protocol \emph{non-interactive} if the server evaluates all polynomials in a single parallel call.
\end{definition}
We ignored here the time it takes the client to encrypt and decrypt, because it is a property of the underlying encryption scheme and not the search algorithm we provide. To be more precise the client time includes also the time to encrypt $\ell$ and decrypt $y$, which require computing $k=O(\log^2 m)$ encryption/decryptions (each for distinct plaintext moduli).


\paragraph{Security guarantee.}
The server sees only encrypted values $\enc{\cost}$, $\enc{\ell}$ (and any values it computes from them, including the output $\enc{f(\cost,\ell)}$), while having no access to a decryption-key or any secret information. The security guaranty is therefore as provided by the underlying encryption. For the aforementioned schemes, the security achieved (under standard cryptographic assumption) is that of \emph{semantic security}~\cite{goldwasser1984probabilistic}, which is the golden standard in cryptography, saying essentially that seeing and manipulating the ciphertexts reveals no new information on the underlying plaintext data (beyond an upper bound on the data size).

\paragraph{Realization of algorithms by polynomial. }
For clarity of the presentation we usually do not describe the polynomial explicitly. Instead, for each polynomial, we suggest algorithm $\textsc{Alg}_{m,p,..}(x,y,\ldots)$ that can be easily implemented (``realized") as a polynomial. The variables of the polynomial correspond to the input of the algorithm, represented as a concatenation $(x,y,\ldots)$ of the input variables. The evaluation of the polynomial corresponds to the output of the algorithm. Parameters such as $m,p$ above are not part of the input nor the output, but part of the polynomial definition itself.
For example, $p$ may be an integer so that the sum and product operations in evaluating the polynomial are modulo $p$, and $m$ may be the input length.
Algorithm $\textsc{Alg}(\cdots)$ with no subscript parameters is assumed to be run on a RAM-machine (``the client"). It may execute commands that cannot be evaluated via a low-degree polynomial. It is called non-interactive if it makes at most a single parallel call for evaluation of polynomials (on ``the server").

\subsection{Our Main Result}

We aim to securely solve the search problem with overall running time poly-logarithmic in the length $m$ of $\cost$, similarly to the running time of binary search on a sorted (non-encrypted) array, i.e., our question is:
\[
  \begin{split}
& \text{\bf Can we solve the Secure Search problem in time that is poly-logarithmic in the array size? }\\
\end{split}
\]
In this paper we answer this question affirmably (see details and proof in Theorem~\ref{lem:FirstPositive-ClientSide}):

\begin{theorem}[Secure Search]\label{thm:secure-search}
There exists a non-interactive protocol that securely solves the search problem on
$\cost\in\set{0,\ldots,r-1}^m$ and $\ell\in\set{0,\ldots,r-1}$,
%
in overall time polynomial in $\log m$ and $\log r$. 
\end{theorem}

%% file: sections/spirit-construction.tex
\section{First-Positive via SPiRiT Sketch} \label{sec:spirit}\label{sec:spirit-construction}
%

The first step in our solution to the search problem is to reduce the input $\cost$ to a binary indicator whose $i$th entry is $1$ if and only if $\cost(i)=\ell$. The search problem then reduces to finding the first positive entry in this indicator vector,
as defined below.

\begin{definition}[First positive index]\label{def:firstpositive}
Let $x=(x_1,\ldots,x_m)\in [0,\infty)^m$ be a vector of $m$ non-negative entries. The \emph{first positive index of $x$} is the smallest index $i^*\in[m]$ satisfying that $x_{i^*}> 0$, or $i^*=0$ if $x=(0,\ldots,0)$.
\end{definition}

The main technical result of this paper is an algorithm for computing the first positive index securely on the server (i.e., via low-degree polynomials). For simplicity, in this section we first assume that Algorithm~\ref{alg:FirstPositive-Server} $\sfirstindex_{m,p,r}(\cost)$ runs on a RAM machine, over real numbers (ring of size $p=\infty$ in some sense). This result is of independent interest, with potential applications as explained in sketch literature over reals or group testing, e.g. to handle streaming data. In the next sections, we introduce more tools to show how to realized the same algorithm by a low-degree polynomial over a ring of size $p\geq1$. For this implementation, the matrices of SPiRiT remain essentially the same, but the implementation of $\isPositive_{p,t}(\cdot)$ for $p<\infty$ will be changed; see Algorithm~\ref{alg:ispositive}.

Algorithm~\ref{alg:FirstPositive-Server} provides a construction scheme for computing the first-positive index $i^*\in[m]$ of a non-negative input vector $\cost$ of length $m$. 
The output is the binary representation $b$ of the desired index minus one, $i^*-1$, when $i^*>0$, and it is $b=(0,\ldots,0)$ when $i^*=0$. (We comment that when the output is $b=(0,\ldots,0)$ there is an ambiguity of whether $i^*=1$ or $i^*=0$; this is easily resolved by setting $i^*=1$ if $\cost(1)>0$, and $i^*=0$ otherwise). The algorithm computes the composition of operators: $S\circ P \circ i \circ R \circ i \circ T$. Here, $S, P, R$ and $T$ are matrices (sometimes called sketch matrices), and $i(\cdot)=\isPositive_{\infty,t}(\cdot)$ is an operator that gets a vector $x$ of length $t$ and returns (binary) indicator vector whose $i$th entry is $1$ if and only if $x(i)\neq 0$ (where $t=2m-1$ in its first use, and $t=m$ in the second).

\begin{algorithm}
    \caption{$\sfirstindex_{m,p,r}(\cost)$\label{alg:FirstPositive-Server}}
  {\begin{tabbing}
\textbf{Parameters:} \=Three integers $m,p,r\geq 1$. In Section~\ref{sec:spirit} only, we use $p=\infty$. Later, $p$ is the ring size.\\
\textbf{Input:} \>A vector $\cost\in [0,\infty)^{m}$ whose first positive index is $i^*\geq 1$; See Definition~\ref{def:firstpositive}.\\
  \textbf{Output:} \>Binary representation $b\in\br{0,1}^{\lceil \log_2m\rceil}$ of $i^*-1$.
   \end{tabbing}}
  \vspace{-0.3cm}
  \tcc{See Section \ref{sec:spirit-construction} for the definition of $S,P,R,T$ and $\isPositive_{p,t}$.}
  Set $w \gets T\cdot\cost$ \tcc{$w\in\REAL^{2m-1}$}\label{l11}
  Set $w' \gets \isPositive_{p,2m-1}(w)$ \label{l12}\tcc{$w'\in\br{0,1}^{2m-1}$}
  Set $v \gets R\cdot w'$ \tcc{$v\in\REAL^m$}\label{l13}
  Set $u \gets \isPositive_{p,m}(v)$\label{l14} \tcc{$u\in\br{0,1}^{m}$}
  Set $b \gets (SP)\cdot u$ \label{l15}\tcc{$b=S\circ P\circ i\circ R\circ i\circ T(\cost)$}
  \Return $b$\label{l16}\tcc{$b\in\br{0,1}^{\lceil \log_2m\rceil}$}\label{l:firstpositive-returnb}
\end{algorithm}

\paragraph{Algorithm~\ref{alg:FirstPositive-Server} Overview. }
The input is a vector $\cost$ of length $m$, and the output is the binary representation of $i^*-1$ for $i^*>0$, and of $0$ for $\cost=(0,\ldots,0)$. The parameters $p,r$ will be utilized in later sections, where the algorithm is assumed to be realized by a polynomial (i.e., run on encrypted data on the server): $p\geq1$ will represent the ring of the polynomial (in this section $p=\infty$), and $r\geq1$ will be the maximal value in $\cost$ (here $r=\infty$). In Line~\ref{l11} the input $\cost$ is right multiplied by a $(2m-1)$ by $m$ matrix $T$, called the tree matrix. in Line~\ref{l12}, the resulting vector $w$ is replaced by an indicator vector $w'$ whose $i$th entry is $1$ if and only if $w(i)>0$. In Line~\ref{l13}, $w'$ is left multiplied by an $m$-by-$(2m-1)$ matrix $R$, called Roots matrix. In Line~\ref{l14} we convert the resulting vector $v$ to an indicator value as before. In Line~\ref{l15}, the resulting vector $u$ is left multiplied by the $\lceil\log_2m\rceil$-by-$m$ matrix $SP$. The matrix $SP$ itself is a multiplication of two matrices: a Sketch matrix $S$ and a Pairwise matrix $P$.

In the rest of the section, we define the components of $S,P,i,R,T$ 
for $m\ge 1$ and prove the correctness of Algorithm~\ref{alg:FirstPositive-Server}. Note that all these components are universal constants that can be computed in advance; see Definition~\ref{global}. Without loss of generality, we assume that $m$ is a power of $2$ (otherwise we pad the input $\cost$ by zero entries).

The first matrix $S$ is based on the following definition of a sketch matrix .
\begin{definition}[Sketch matrix.\label{sketch}]
Let $s,k\geq1$ be integers. A binary matrix $S\in\br{0,1}^{k\times m}$ is called an \emph{$(s,m)$-sketch matrix}, if the following holds. There is an algorithm $\Decodesketch$ that, for every vector $y\in\REAL^k$, returns a binary vector $\iz=\Decodesketch(y)\in\br{0,1}^m$ if and only if $y=S\cdot \iz$. 
The vector $y\in\REAL^k$ is called the \emph{$s$-coreset} of the vector $\iz$.
\end{definition}

\paragraph{The Sketch $S\in {\br{0,1}}^{\log m\times m}$}~is a $(1,m)$-sketch matrix as in Definition~\ref{sketch}. Its right multiplication by a binary vector $t=(0,\cdots,0,1,0,\cdots,0)\in \br{0,1}^m$ which has a single ($s=1$) non-zero entry in its $k$th coordinate, yields the binary representation $y=St\in\br{0,1}^{\log m}$ of $k$.
A $(1,m)$-sketch matrix $S$ can be easily implemented by setting each column $k=\set{1,\ldots,m}$ to be the binary representation of $k-1$. More generally and for future work where we wish to search for $s\geq 2$ desired indices, an $(s,m)$-sketch matrix should be used. Efficient construction of an $(s,m)$ sketch matrix for $s\geq2$ is more involved and is explained in e.g.~\cite{indyk2010efficiently}.
\removed{See Section~\ref{futurework}.}

\paragraph{Pairwise matrix $P\in \br{-1,0,1}^{m\times m}$}~is a matrix whose right multiplication by a given vector $u\in\REAL^m$ yields the vector $t=Pu$ of pairwise differences between consecutive entries in $u$, i.e., $t(k)=u(k)-u(k-1)$ for every $k\in\br{1,\cdots,m}$ and $t(1)=u(1)$. Hence, every row of $P$ has the form $(0,\cdots,-1,1,\cdots,0)$. For example, if $m=7$ and $u=(0,0,0,1,1,1,1,1)$ then $t=Pu=(0,0,0,1,0,0,0)$. More generally, if $u$ is a binary vector that represents a step function, then $t$ has a single non-zero entry. Indeed, this is the usage of the Pairwise sketch in $\spirit$.

\medskip\noindent\textbf{The operator} $\isPositive_{\infty,t}(\cdot)$, or $i(\cdot)$ for short, gets as input a vector $v\in\REAL^{t}$, and returns a binary vector $u\in\zo^{t}$ 
where, for every $k\in[t]$, we have $u(k)=0$ if and only if $v(k)=0$. In Section~\ref{sec:ispositive} we define $\isPositive_{p,t}(\cdot)$ for every integer $p\geq1$.
%

To define the next matrix, we define tree representation of a vector $x$, based on the common array representation of a tree as defined in~\cite{cormen2009introduction}. See Fig.~\ref{fig:tree} for a simple intuition of the tedious definitions for the matrix $T$, $R$, and the tree representation of $x$.
\begin{figure}[ht!]
\centering
\includegraphics[width=90mm]{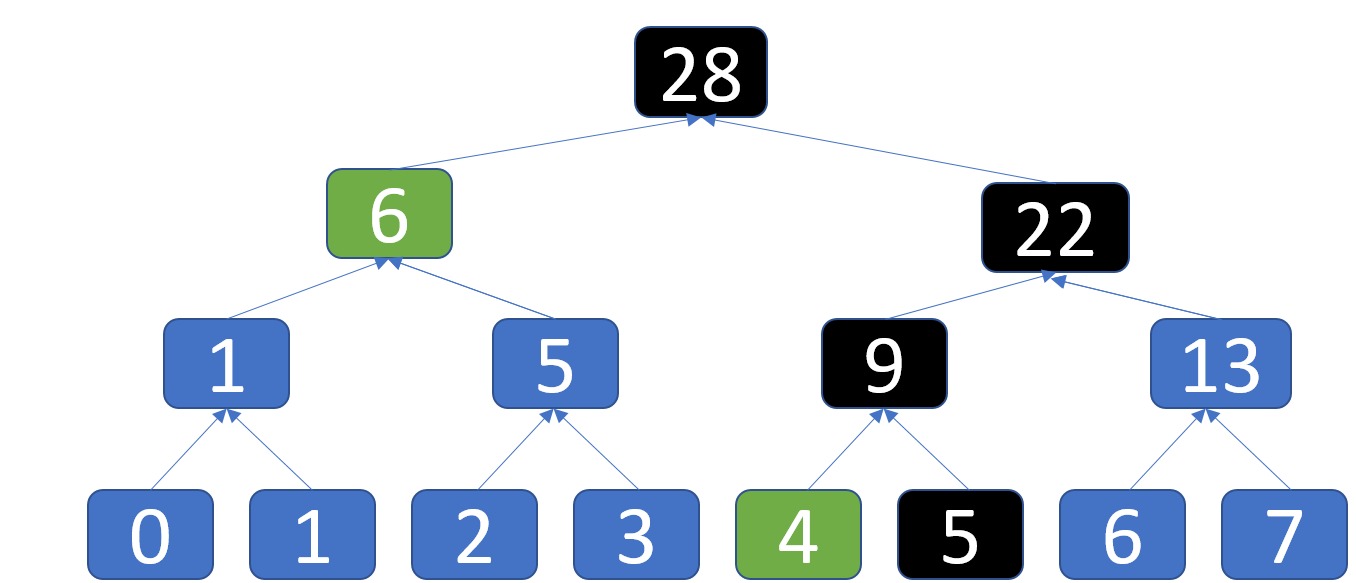}
\caption{
\small The tree representation $\tree(x)$ of the vector $x=(0,1,2,\cdots,7)$ for $m=8$ leaves. Its array representation is the vector $w=w_x=(28,6,22,1,5,9,13,0,1,2,3,4,5,6,7)$, which is the result of scanning the rows of $\tree$ from top to bottom. The Tree matrix $T$ has the property that $w_x=Tx$ for every $x$. The label of each inner node is the sum of its children's labels. The sum of leaves' labels up to the $i=5th$ leaf from the left (labeled '4') is given by $v(5) = x(1)+x(2)+x(3)+x(4)+x(5) = 0+1+2+3+4=10$. More generally, $v=v_x=(0,1,3,6,10,15,21,28)$. The root matrix $R$ has the property that $RTx=Rw=v$. \newline To construct sparse $R$, we prove that every entry of $v$ can be computed using only $O(\log m)$ labels. For example, $2$ labels for its $i=5th$ entry $v(5)$ as follows. First identify all the roots (ancestors) of the $i+1=6th$ leaf: 5,9,22,28 (black in the figure). Among these ancestors, select those who are right children: labels 5 and 22 in the figure. Finally, sum the labels over the left siblings of these selected ancestors: 4 and 6 (in green in the figure) to get the desired sum. Indeed, $v(5) = 4+6=10$.
\label{fig:tree}}
\end{figure}

\begin{definition}[tree/array representation]\label{tree}
Let $x=(x(1),\cdots,x(m))\in\REAL^{m}$ be a vector. The \emph{tree representation} $\tree(x)$ of $x$ is the full binary tree of depth $\log_2 m$, where each of its node is assigned a label as follows. The label of the $i$th leftmost leaf is $x(i)$, for every $i\in[m]$. The label of each inner node of $\tree(x)$ is defined recursively as the sum of the labels of its two children. 

The \emph{array representation} of $x$ is the vector $w=(w(1),\cdots,w(2m-1))\in \REAL^{2m-1}$, where $w(1)$ is the label in the root of $\tree(x)$, and for every $j\in[m-1]$ we define $w(2j)$, $w(2j+1)$ respectively to be the labels of the left and right children of the node whose label is $w(j)$; see Fig.~\ref{fig:tree}.
\end{definition}
In particular, the last $m$ entries of $w$ are the entries of $x=\big(w(m),\cdots,w(2m-1)\big)$.

\paragraph{Roots sketch $R\in \br{0,1}^{m\times (2m-1)}$ }is a binary matrix such that
 \begin{enumerate}
\item Each row of $R$ has $O(\log m)$ non-zero entries.
\item For every tree representation $w\in\REAL^{2m-1}$ of a vector $x\in\REAL^m$ we have that $v=Rw\in \REAL^m$ satisfies for every $j\in[m]$ that $v(j)$ is the sum of entries $x(1),x(2),\ldots,x(j)$ of $x$, i.e.,
$v(j)=\sum_{k=1}^{j} x(k).$
\end{enumerate}
In particular, and for our main applications, if $x$ is non-negative, then $v(j)=0$ if and only if $x(1)=x(2)=\cdots=x(j)=0$. That is, $v(j)$ tells whether $j$ is equal-to or larger-than the index of the first positive entry in $x$. Below is our implementation of the matrix $R$, which also explains its name.

\paragraph{Implementation for the Roots matrix $R$\label{RR}}
Let $x\in[0,\infty)^m$ and $w\in\REAL^{2m-1}$ be its tree representation. We need to design a row-sparse matrix $R$ such that if $v=Rw$ then $v(j)=\sum_{k=1}^{j} x(k)$.  
Our main observation for implementing such a row-sparse matrix $R$ is that $v(j)$ can be computed by summing over only $O(\log m)$ labels in the tree $\tree(x)$; specifically, the sum is over the labels of the set $\lop_{\tree}(j+1)$ of
left-siblings of the ancestors (roots) of the $(j+1)$th leaf; see Fig.~\ref{fig:tree}. By letting $\lop(j+1)$ denote the corresponding indices of this set in the tree representation $w=Tx$ of $x$, we have
\[
v(j)=\sum_{k\in\lop(j+1)} w(k)
\]
This equality holds because the left-siblings in $\lop_{\tree}(j+1)$ partition the leaves $x(1),\ldots,x(j)$ into 
disjoint sets (with a set for each left-sibling), so the sum over each of these sets is the value of the corresponding labels in $w$.

The indices of these left-siblings is formally defined as follows.
\begin{definition}[$\parents,\lop$\label{def:lop}]
Let $x=(x(1),\cdots,x(m))\in\REAL^{m}$, and $\tree=\tree(x)$ be the tree-representation of $x$. Consider the $j$th leftmost leaf in $\tree$ for $j\in [m]$. We define $\parents_{\tree}(j)$ to be the set of ancestors nodes of this leaf. We denote by $\lop_{\tree}(j)$ the union of left-sibling nodes over each ancestor in $\parents_{\tree}(j)$ who is a right-sibling of its parent node. Let $w\in\REAL^{2m-1}$ be the array representation of $\tree$. We denote by $\parents(j),\lop(j)\subseteq [2m-1]$ the set of indices in $w$ that correspond to $\parents_{\tree}(j)$, and $\lop_{\tree}(j)$ respectively.
\end{definition}

From the above discussion we conclude the following implementation of $R$ satisfies the definition for the Roots sketch matrix.
\begin{lemma}[Roots sketch]\label{lem:roots}
Let $R\in\zo^{m\times(2m-1)}$ such that for every $j\in[m]$ and $\ell\in [2m-1]$ its entry in the $j$th row and $\ell$th column is $R(j,\ell) = 1$ if $\ell\in\lop(j+1)$, and $R(j,\ell) = 0$ otherwise;  see Definition~\ref{def:lop} and Fig.~\ref{tree}. Then $R$ is a Roots sketch matrix as defined above.
\end{lemma}

\paragraph{The Tree matrix $T\in \br{0,1}^{(2m-1)\times m}$ }is a binary matrix that, after right multiplication by a vector $x\in\REAL^{m}$, returns its tree representation $w\in\REAL^{2m-1}$; see Definition~\ref{tree}.
The value $w(i)$ is a linear combination of entries in $x$, specifically, the sum of the labels in the leaves of the sub-tree that is rooted in the inner node corresponding to $w(i)$. Hence, such a matrix $T$ that satisfies $w=Tx$ can be constructed by letting each row of $T$ corresponds to a node $u$ in $\tree(x)$, and has $1$ in every column $j$ so that $u$ is the the ancestor of the $j$th leaf. Note that this unique matrix $T$ can be constructed obliviously and is independent of $x$.

%

\begin{theorem}[First-positive over non-negative reals\label{lem1}]\label{lem:firstpositive-reals}
Let $r,m\geq1$ be integers, and $\cost\in[0,\infty)^m$ be a non-negative vector.
Let $y$ be the output of a call to $\sfirstindex_{m,\infty,r}(\cost)$; see Algorithm~\ref{alg:FirstPositive-Server}.
Then $y$ is the binary representation of the first positive index of $\cost$. See Definitions~\ref{def:sketch-firstpositive} and~\ref{def:firstpositive}.
\end{theorem}
\begin{proof}
  Proof appears in Appendix~\ref{sec:analysis-first-prositive-reals}.
\end{proof}

%
\section{Secure Search}
%

In this section we use the result of the previous section to compute a $k$-search coreset for a given vector $\cost\in\br{0,\cdots,r-1}^m$ via polynomials of low-degree.

\subsection{Secure SPiRiT}\label{sec:ispositive}

To run SPiRiT over low-degree polynomials we implement $\isPositive_{p,t}$ from the previous section via a polynomial over a ring of size $p<\infty$; see Appendix~\ref{ispossec}.
%

Using this implementation of $\isPositive_{p,t}(x)$, we can realize SPiRiT in Algorithm~\ref{alg:FirstPositive-Server} via a polynomial of degree $(p-1)^2$ \anote{I changed $p$ to $(p-1)^2$} over a ring $\ZZ_p$. Since we aim for degree $(\log m)^{O(1)}\ll m$, we take $p=(\log m)^{O(1)}$. However, for such small $p$, the correctness of SPiRiT no longer holds (\eg, since summing $m$ positive integers modulo $p$ may result in $0$ over such a small modulus $p$). Fortunately, we can still prove correctness under the condition specified in Lemma~\ref{lem:spirit-sufficient-condition} below. In the next sections we design algorithms that ensure this condition is satisfied.

\begin{definition}[$A$-correct]\label{def:T-A-correct}
Let $A$ be a set of integers. An integer $p\geq1$ is \emph{$A$-correct} if for every $a\in A$ we have $a=0$ if-and-only-if $(a \mod p)=0$.
\end{definition}

%

\begin{lemma}[Sufficient condition for success of $\sfirstindex_{m,p,r}$.]
\label{lem:spirit-server}
\label{lem:spirit-sufficient-condition}
    Let $m,p,r\geq 1$ be integers, $\cost\in\set{0,\cdots,r-1}^m$ a vector whose first positive index is $i^*\in[m]\cup\set{0}$, and
    $
    A=\sett{(T\cdot\cost)(j)}{j\in\parents(i^*)}
    $
    for $\parents(\cdot)$ and the Tree matrix $T$ as defined in Section~\ref{sec:spirit}.
    If $p$ is $A$-correct, 
    then the output $b$ returned from $\sfirstindex_{m,p,r}(\cost)$ (see Algorithm~\ref{alg:FirstPositive-Server})
    is the binary representation of $i^*-1$ (and $b=0^{\log m}$ if $i^*=0$).
    Moreover, $\sfirstindex_{m,p,r}(\cost)$ can be realized by a non-interactive parallel call to $\log m$ polynomials (one for each output bit), each of log-size $O(\log m)$ and degree $O(p^2)$.
\end{lemma}
\begin{proof}
    Fix $i^*\in[m]$ (the case $i^*=0$ is trivial, details omitted). We first show that if $p$ is $A$-correct then $\sfirstindex_{m,p,r}$ returns the binary representation $b^*$ of $i^*-1$.
    Denote $x=\cost$; and denote by $i_p(),i_\REAL()$ the $\isPositive$ operators used in $\sfirstindex_{m,p,r}$ and $\sfirstindex_{m,\infty,r}$ respectively. Namely, these operators, given an integer vector map its entries to binary values, where $i_p()$ maps to $0$ all multiples of $p$, and $i_\REAL$ maps to $0$ only on the real number zero.
    Let
    \[
    u(j) = i_p(R(i_p(Tx\mod p))\mod p)(j) = i_p\left(\sum_{k\in\lop(j+1)}(i_p(Tx\mod p))\right)(j)
    \]
    (where the last equality is by construction of $R$).
    We show below that for all $j\in[m]$, the following holds:
    \begin{itemize}
        \item if $j<i^*$, then $u(j)=0$ (see Claim~\ref{claim:suff-cond-low-i}); and
        \item if $j\ge i^*$ and $p$ is $A$-correct for $A=\parents(i^*)$, then $u(j)=1$ (see Claim~\ref{claim:suff-cond-high-i}).
    \end{itemize}
    This implies (by Claim~\ref{claim:SP-p}) that a call to $\sfirstindex_{m,p,r}(x)$ returns the binary representation $b^*$ of $i^*-1$.

    We next analyze the complexity of the polynomial $f()$ realizing $\spirit_{p,m,r}$. $f()$ is the composition of polynomials realizing the matrices $(SP)$, $R$ and $T$, and polynomials realizing the two evaluations of the operator $i_p()=\isPositive_{p,t}()$,
    with degree and size $1$ and $O(m^2)$ for the former (where size is simply the number of matrix entries), and $p-1$ and $1$ for the latter.
    The degree of $f$ is therefore $1\cdot(p-1)\cdot 1\cdot(p-1)\cdot 1 = (p-1)^2$, and its size $O(m^2\cdot 1\cdot m^2\cdot 1\cdot m^2)=O(m^6)$ implying log-size of $O(\log m)$. (We remark that the size bound is not tight; more accurately, we can count only the non-zero entries in the matrices.)

\begin{claim}
\label{claim:suff-cond-low-i}
    $u(j) = 0$ for all $j<i^*$.
\end{claim}
\begin{proof}
  Fix $j<i^*$. We first show that $Tx(k) = 0$ for all $k\in\lop(j+1)$. For this purpose observe that
  \[
  Tx(k) = \sum_{i \text{ s.t. }k\in\parents(i)}x(i) \le \sum_{i=1}^j x(i),
  \]
  where the equality is by definition of the tree matrix $T$, and the inequality follows from $k\in\lop(j+1)$ being an ancestor only of (a subset of) the first $j$ leaves, and $x$ being non-negative. The above implies that $Tx(k)=0$ because $x(1)=\ldots=x(j)=0$ for $j$ smaller than the first positive index $i^*$.
  We conclude therefore that
  $i_p((Tx\mod p)(k)) = 0$,
  and the sum of these values 
  over all $k\in\lop(j+1)$ is also zero:
  \[
  (R(i_p(Tx\mod p))\mod p)(j) = \sum_{k\in\lop(j+1)}(i_p(Tx \mod p))(k) = 0.
  \]
  implying that
  $u(j) = i_p(R(i_p(Tx\mod p))\mod p)(j) = 0.$
\end{proof}

\begin{claim}
\label{claim:suff-cond-high-i}
    Suppose $p$ is $A$-correct for the set $A = \parents(i^*)$. Then $u(j) = 1$ for all $j\ge i^*$.
\end{claim}
\begin{proof}
  Fix $j\ge i^*$. The key observation is that the intersection $\parents(i^*)$ and $\lop(j+1)$ is non-empty:
  \[
  \exists k^* \in \parents(i^*)\bigcap \lop(j+1).
  \]
  The above holds because the left and right children $v_L,v_R$ of the deepest common ancestor of $i^*$ and $j+1$ must be the parents of $i^*$ and $j+1$ respectively (because $i^* < j+1$), implying that $v_L$ is both an ancestor of $i^*$ and a left-sibling of the ancestor $v_R$ of $j+1$. Namely, $v_L$ is in the intersection of $\parents(i^*)$ and $\lop(j+1)$.
  Now, since $k^* \in\parents(i^*)=A$ then,
  \[
  Tx(k^*) = \sum_{j\text{ s.t. }k^*\in\parents(j)}x(j) \ge x(i^*) \ge 1,
  \]
  %
  implying by $A$-correctness of $p$ that,
  \[
  (i_p(Tx\mod p))(k^*) = 1.
  \]
  Therefore, since $k^*\in\lop(j+1)$, we have the lower bound:
  \[
  (R(i_p(Tx\mod p)))(j) = \sum_{k\in\lop(j+1)}(i_p(Tx\mod p))(k) \ge (i_p(Tx\mod p))(k^*) = 1.
  \]

  Conversely, since the above is a sum over at most $\card{\lop(j+1)} \le \log m < p$ bits (where the first inequality is a bound on the depth of a binary tree with $m$ leaves, and the second inequality follows from the definition of $\Primes_{m,s}$; see Definition~\ref{def:primes-m-s}), we have the upper bound:
  \[
  (R(i_p(Tx\mod p)))(j) \le \log m < p.
  \]

  We conclude therefore that the above is non-zero even when reduced modulo $p$:
  \[
  u(j) = i_p(R(i_p(Tx\mod p))\mod p)(j) = 1.
  \]
\end{proof}

\begin{claim}\label{claim:SP-p}
  Let $u=(0,\ldots,0,1\ldots,1)$ a length $m$ binary vector accepting values $u(1)=\ldots=u(i^*-1)=0$ and values $u(i^*)=\ldots=u(m)=1$. Then $(SP\cdot u \mod p)$ is the binary representation of $i^*-1$.
\end{claim}
\begin{proof}
  Multiplying $u$ by the pairwise difference matrix $P\in\set{-1,0,1}^{m\times m}$ returns the binary vector $t = Pu\mod p$ in $\zo^m$ defined by $t(k) = u(k)-u(k-1)$ (for $u(0)=0$). This vector accepts value $t(i^*)=1$ and values $t(i)=0$ elsewhere.
  Multiplying $t$ by the sketch $S\in\zo^{\log m\times m}$ returns the binary vector $y=St\mod p$ in $\zo^{\log m}$ specifying the binary representation of $i^*-1$.
  We conclude that $(SP\cdot u\mod p)$ is the binary representation of $i^*-1$.
\end{proof}
\end{proof}

%
\subsection{Search Coreset's Item}\label{sec:SearchCoresetItem}
%

As explained in Lemma~\ref{lem:spirit-server}, running SPiRiT on the server would yield a result that may not be the desired first positive index of the input vector, but will serve as an item in a $k$-coreset for search. Moreover, in the Search problem we are interested in the entry that consists of a given lookup value $\ell$, and not on the first positive index. In this section we handle the latter issue.

\begin{algorithm}
    \caption{$\oneWitness_{m,p,r}(\cost, \ell)$\label{alg:oneWitness}}
  {\begin{tabbing}
\textbf{Parameters:} \=Three integers $m,p,r\geq 1$.\\
\textbf{Input:} \>A vector $\cost\in\set{0,\ldots,r-1}^{m}$ and a lookup value $\ell\in\set{0,\ldots,r-1}$ to search for.\\
  \textbf{Output:} \>An item of a search coreset that is constructed in Algorithm~\ref{alg:FirstPositive-ClientSide1}; see Lemma~\ref{lem:searchserver}.
   \end{tabbing}}
  \vspace{-0.3cm}
\tcc{$\forall i: \ind(i)=1$ if $\cost(i)=\ell$, and $\ind(i)=0$ otherwise}
      $\ind\gets \tobinary_{m,p,r}(\cost,\ell)$\label{l:search-coreset-tobinary}
      \tcc{see Algorithm~\ref{alg:ToBinaryExact}
}
  Set $b_p \gets \sfirstindex_{m,p,r}(\ind)$
  \tcc{See Algorithm~\ref{alg:FirstPositive-Server}}
  \label{l:search-coreset-spirit}
  \Return $b_p$
\end{algorithm}

\paragraph{Overview of Algorithm~\ref{alg:oneWitness}. }
For our main application the algorithm is to be run on encrypted data on the server side (i.e., realized by a polynomial). The fixed parameters 
are the integers $m,p,r\geq 1$, where $\ZZ_p$ is the ring in which all operations are computed, $m$ is the size of the expected input array, and $\set{0,\ldots,r-1}$ is the set of possible values for each entry. The input is $\cost$ of $m$ integers in $\set{0,\cdots,r-1}$ with an additional integer $\ell\in\set{0,\ldots,r-1}$ to look for in this array. The output of $\oneWitness_{m,p,r}(\cost, \ell)$ is a vector $b_p$ of $\log m$ bits that, under the condition on $p$ from Lemma~\ref{lem:spirit-server}, represents the smallest value $i^*-1$, where $\cost(i^*)=\ell$. For example, if $\log m=3$, $p=10$, $\cost=(4,2,3,9,5,4,9,2)$ and $\ell=9$ then $i^*-1=3$ since $\cost(4)=9$, and the output is $b_p=(0,1,1)$ .

In Line~\ref{l:search-coreset-tobinary} 
we construct a length $m$ binary array $\ind$ whose value is $1$ only on the entries where $\cost$ contains the lookup value $\ell$; this is done via Algorithm~\ref{alg:ToBinaryExact} on the appendix.
%
In Line~\ref{l:search-coreset-spirit} 
we return the first such entry, using Algorithm~\ref{alg:FirstPositive-Server}, which in turn calls the $\spirit$ sketch.

%

\begin{lemma}[Sufficient condition for success of $\oneWitness_{m,p,r}$]\label{lem:searchserver}
    Let $m,p,r\geq 1$ be integers, $\cost\in\set{0,\cdots,r-1}^m$ a vector, $\ell\in\set{0,\ldots,r-1}$ a lookup value,
    and $i^*\in[m]$ the smallest index where $\cost(i^*)=\ell$.
    Let
    \[
    A=\sett{(T\cdot\ind)(j)}{j\in\parents(i^*)}
    \]
    where $\ind$ is the vector computed in Line~\ref{l:search-coreset-tobinary} of Algorithms~\ref{alg:oneWitness}, and $\parents(\cdot)$ and the Tree matrix $T$ are as defined in Section~\ref{sec:spirit}.
    If $p$ is $A$-correct, 
    then the output $b_p$ returned from $\oneWitness_{m,p,r}(\cost, \ell)$ (see Algorithm~\ref{alg:oneWitness})
    is the binary representation of $i^*-1$ (and $b_p=0^{\log m}$ if $\ell$ does not appear in $\cost$).
    Moreover, $\oneWitness_{m,p,r}(\cost, \ell)$ can be realized by a non-interactive parallel call to $\log m$ polynomials (one for each output bit), each of log-size $O(\log m+\log r)$ and degree $O(p^2\log r)$.
\end{lemma}
\begin{proof}
  The proof follows immediately from Lemmas~\ref{lem:spirit-sufficient-condition} and \ref{lem:tobinary}, analyzing the correctness and complexity of $\sfirstindex$ and $\tobinary$ respectively.\anote{elaborate if there's time}
\end{proof}

Note that if $p \geq m + 1$, then $p$ is $A$-correct and $b_p$ is the desired output: the binary representation of $i^*-1$, where $\cost(i^*)=\ell$ is the first occurrence of $\ell$ in $\cost$. Since in this case every item $a\in A$ is upper bounded by the length $m$ of $\ind\in\zo^m$, and thus $a \bmod p = a$. However, such a ring size $p$ would result in a polynomial of a degree that is linear in $m$, whereas we desire for poly-logarithmic dependency. This is why we use $k$-coreset as explained in the next section.

\subsection{Main Algorithm for Secure Search}\label{sec:FirstPositive-Client}

In this section we present our main search algorithm that, given a vector $\cost\in\set{0,\ldots,r-1}^m$, and a lookup value $\ell$, returns the first index of $\cost$ that contains $\ell$. This is done by constructing a $k$-coreset where $k=O(\log^2 m)$. We prove that the algorithm can be implemented in efficient (poly-logarithmic in $m$) overall time. The coreset consists of the outputs of the same polynomial (secure algorithm that run on the server) over different small prime rings sizes. The resulting coreset is computed via a non-interactive parallel call from the client (RAM machine) who combines them to conclude the desired index $i^*$ that contains $\ell$. The set of primes is denoted using the following definition.

\begin{definition}[$\Primes_{m,s}$]\label{def:primes-m-s}
For every pair of integers $m,s\ge 1$, and $b=\ceil{\log m}$ we define $\Primes_{m,s}$ to be the $1 + \ceil{s}\cdot \log_{b} m$ smallest prime numbers that are larger than $b$.
\end{definition}




\begin{algorithm}
    \caption{$\securesearch(\cost,\ell)$ \label{alg:FirstPositive-ClientSide1} }
{\begin{tabbing}
\textbf{Input:\quad} \=A vector $\cost\in \br{0,\cdots,r-1}^{m}$ and $\ell\in\br{0,\cdots,r-1}$.\\
\textbf{Output:} \>The smallest index $i^*$ such that $\cost(i^*)=\ell$\\\>
or $i^*=0$ if there is no such index.
\end{tabbing}}
\vspace{-0.3cm}
    \For {$\each$ $p\in \Primes_{m,\lceil\log m\rceil}$ \label{l:clientFirstPositive-loop}
    \label{l:securesearch-Primes}
    (See Definition~\ref{def:primes-m-s})}    {
        $b_p\gets \oneWitness_{m,p,r}(\cost, \ell)$
        \label{l:securesearch-b-p} 
        \tcc{ see Algorithm~\ref{alg:oneWitness}}
        Set $i_p\in[m]$ such that $b_p$ is the binary representation of $i_p-1$\label{l43}\\
    }
    Set $C\gets \set{i_p}_{p\in\Primes_{m,\log m}}$ 
    \label{l:securesearch-C} \\
    \tcc{$C$ is a $k$-search coreset, for $k=\card{\Primes_{m,\log m}}$, by the proof of Theorem~\ref{lem:FirstPositive-ClientSide}}
    Set $i^*\gets $ the smallest index $i^*$ in $C$ such that $\cost(i^*)=\ell$, or $i^*=0$ if there is no such index.
    \label{l:securesearch-i}\\
    \Return $i^*$
%
\end{algorithm}

\paragraph{Overview of Algorithm~\ref{alg:FirstPositive-ClientSide1}.}
%
The algorithm runs on the client side (RAM machine) and its input is a vector $\cost\in\set{0,\ldots,r-1}^m$ together with a lookup value $\ell$.
In Line~\ref{l:securesearch-Primes}, 
$\Primes_{m,\lceil \log m\rceil}$ is a set of primes that can be computed in advance.
In Line~\ref{l:securesearch-b-p}, 
the algorithm makes the following single non-interactive parallel call to the server. For every prime $p\in\Primes_{m,\lceil \log m\rceil}$, 
Algorithm~\ref{alg:oneWitness} is applied and returns a binary vector $b_p\in\zo^{\log m}$. We prove in Theorem~\ref{lem:FirstPositive-ClientSide} that the resulting set $C$ of $k=|\Primes_{m,\lceil \log m\rceil}|$ binary vectors is a $k$-search coreset; see Definition~\ref{def:sketch-firstpositive}. In particular, one of these vectors is the binary representation of $i^*-1$, where $i^*$ is the smallest entry that contains the desired lookup value $\ell$.
In Line~\ref{l:securesearch-i} 
the client checks which of the coreset items indeed contain the lookup value. The smallest of these indices is then returned as output.
We remark that if $\cost$ cannot be accessed (for example, if $\cost$ is maintained only on the server), we modify Algorithm~\ref{alg:oneWitness} to return both $b_p$ and $\cost(i_p)$ (see details in Section~\ref{sec:generic-search}). 

The following theorem is the main result of this paper, and suggests an efficient solution for the problem statement in Definition~\ref{def:search}.
\begin{theorem}[Secure Search in poly-logarithmic time]\label{lem:FirstPositive-ClientSide}
    Let $m,r\ge1$ be integers, and $\cost\in\set{0,\ldots,r-1}^m$. Let $i^*$ be the output returned from a call to $\securefirstindex(\cost,\ell)$; see Algorithm~\ref{alg:FirstPositive-ClientSide1}. Then
     \[
           i^* = \min\sett{i\in[m]}{\cost(i)=\ell}.
     \]
     Furthermore, Algorithm~\ref{alg:FirstPositive-ClientSide1} can be computed in client time $O(\log^3 m)$ and a non-interactive parallel call to $O(\log^3 m)$ polynomials, each of log-size $O(\log m + \log r)$ and degree $O(\log^4 m \log r)$. This results in an overall time of $O(\log^4 m \log r)$ (see Definition~\ref{def:secure-search}). 
\end{theorem}


%
%

\begin{proof}[of Theorem~\ref{lem:FirstPositive-ClientSide}]

    \paragraph{Correctness.}
    Fix $\cost\in\set{0,\ldots,r-1}^m$ and $\ell\in\set{0,\ldots,r-1}$.
    Let $i^* = \min\sett{i\in[m]}{\cost(i^*)=\ell}$, 
    and let $A=\parents(i^*)$ (see Definition~\ref{def:lop}).
    Let $b_p, i_p, C$ as defined in Algorithm~\ref{alg:FirstPositive-ClientSide1}; that is, $b_p\gets \oneWitness_{m,p,r}(\cost, \ell)$ is the binary representation of $i_p-1$ and $C=\set{i_p}_{p\in\Primes_{m,\log m}}$.

    We show that the set $C$ contains desired output $i^*$:
    By Lemma~\ref{lem:spirit-sufficient-condition}, if $p$ is $A$-correct (see Definition~\ref{def:T-A-correct}) then $b_p$ is the binary representation of $i^*-1$.
    By Lemma~\ref{lem:existentialCRT} since $\card A \le \log m$, then there exists a $p^*\in\Primes$ that is $A$-correct.
    We conclude therefore that $i^*\in C$.

    Finally, observe that as $C\subseteq[m]$, then $i^*$ being the smallest index in $[m]$ for which $\cost(i^*) = \ell$, implies that $i^*$ is also the smallest index in $C$ satisfying that $\cost(i^*) = \ell$ (or $i^*=0$ if no such index exists). Thus, the output $\min\sett{i\in C}{\cost(i)=\ell}$ is equal to 
    $i^*$.%
    %

    \paragraph{Complexity.}
    The client executing Algorithm~\ref{alg:FirstPositive-ClientSide1} first makes a non-interactive parallel call to compute $\sfirstindex_{p,m,r}$ for the $k=\card{\Primes_{m,\ceil{\log m}}}$ values $p\in\Primes_{m,\ceil{\log m}}$. Each such computation calls $O(\log m)$ polynomials (one polynomial for each bit of $b_p$), each of degree $O(p^2\log r)$ and log-size $O(\log m + \log r)$ (see Lemma~\ref{lem:searchserver}). Next the client runs in $O(k)$ time to process the returned values. 
    By Definition~\ref{def:primes-m-s}, $k = O(\log^2 m /\log\log m) = o(\log^2 m)$; By Lemma~\ref{lem:existentialCRT}, the magnitude of the primes $p\in\Primes_{m,\ceil{\log m}}$ is upper bounded by $p = O(\log^2 m)$.
    We conclude that Algorithm~\ref{alg:FirstPositive-ClientSide1} can be computed in client time is $O(k\log m) = o(\log^3 m)$, and a non-interactive parallel call to compute $k\log m = o(\log^3 m)$ polynomials, each degree $O(\log^4 m \log r)$ and log-size $O(\log m + \log r)$.
\end{proof}


Lemma~\ref{lem:existentialCRT} below states that for every set $A$ of size at most $s$, there exists a prime $p^*\in\Primes_{m,s}$ so that $p^*$ is $A$-correct.

\begin{lemma}[existential CRT sketch]\label{lem:existentialCRT}
For every integers $m,s\ge 1$, the set $\Primes_{m,s}$ (see Definition~\ref{def:primes-m-s}) satisfies the following properties.
\begin{enumerate}
  \item For every set $A\subseteq\set{0,\ldots,m}$ of size at most $s$, there exists $p^*\in\Primes_{m,s}$ that is $A$-correct.
  \item The magnitude of $p\in\Primes_{m,s}$ is upper bounded by $p = O(s\log m)$.
\end{enumerate}
\removed{\footnote{
For clarity we spell out what the above $O()$ notation means. There exists constants $c_1,c_2, c_3 > 0$ such that for every sufficiently large $m$, the following holds: For every $s\le c_1 \log m$, the set $\Primes_{m,s}$ is of size at most $\card{\Primes_{m,s}} \le c_2 \log^2 m/\log\log m$, and the magnitude of the primes $p$ in $\Primes_{m,s}$ is at most $p < c_3 \log^2 m$.
}}
\end{lemma}

\begin{proof}[of Lemma~\ref{lem:existentialCRT}]
  Fix $A\subseteq\set{0,\ldots,m}$ of cardinality at most $s$. We first prove that there exists $p^*\in\Primes_{m,s}$ that is $A$-correct.
  Recall that $p^*$ is $A$-correct if for every $a\in A$, $(a\mod p^*)=0$ if-and-only-if $a=0$ (see Definition~\ref{def:T-A-correct}). Clearly if $a=0$ then $(a\mod p^*)=0$. Therefore, it suffices to show that if $a\neq 0$ then $(a\mod p^*)\neq 0$. Namely, it suffices to show that there exists $p^*\in\Primes_{m,s}$ that divides none of the non-zero elements $a$ in $A$.
  %
  %
  %
  Consider an element $a$ in $A$. Observe that $a\le m$ has at most $\log_b m$ prime divisors larger than $b$, and therefore, at most $\log_b m$ divisors in $\Primes_{m,s}$. Thus, the number of elements $p$ in $\Primes_{m,s}$ so that $p$ is a divisor on an element $a\in A$ is at most $s\cdot\log_b m$. Now since $\card{\Primes_{m,s}} > s\cdot\log_b m$, then by the Pigeonhole principle there exists $p^*\in\Primes_{m,s}$ that divides none of the elements $a\in A$.


  Next we bound the magnitude of the primes $p$ in $\Primes_{m,s}$. For this purpose recall that by the Prime Number Theorem (see, \eg, in \cite{hardy75}) asymptotically we expect to find $x/\ln x$ primes in the interval $[1,x]$. Thus, we expect to find $\f x {\ln x} - \f b{\ln b} = \Omega(\f x {\ln x})$ primes in the interval $[b,x]$, where the last equality holds for every $b=o(x)$. Assign $x = t\ln t$ for $t = 1+s\log_b m$ and $b=\ceil{\log m}$. 
  For sufficiently large $m$ there are $t$ primes larger than $b$ in the interval $[b,t\ln t]$, so all the primes in $\Primes_{m,s}$ are of magnitude at most $p = O(t\ln t) = O(s \log m)$.
  \end{proof}

\begin{remark}
    To construct $\Primes_{m,s}$ we simply take the first $t = 1+s\log_b m$ primes larger than $b$. This surely gives a set of $t$ primes, though for small $t$ these primes might not contained in the interval $[b,t\ln t]$. Nonetheless, the above shows that for sufficiently large $t$, these primes are all of magnitude at most $t \ln t$. This asymptotic statement is captured by the $O()$ notation saying that the primes in $\Primes_{m,s}$ are all of magnitude $O(t\ln t)$. 
\end{remark}

%% file: sections/experimental_results_spirit2.tex
\begin{figure}[!ht]
\centering
\includegraphics[width=90mm]{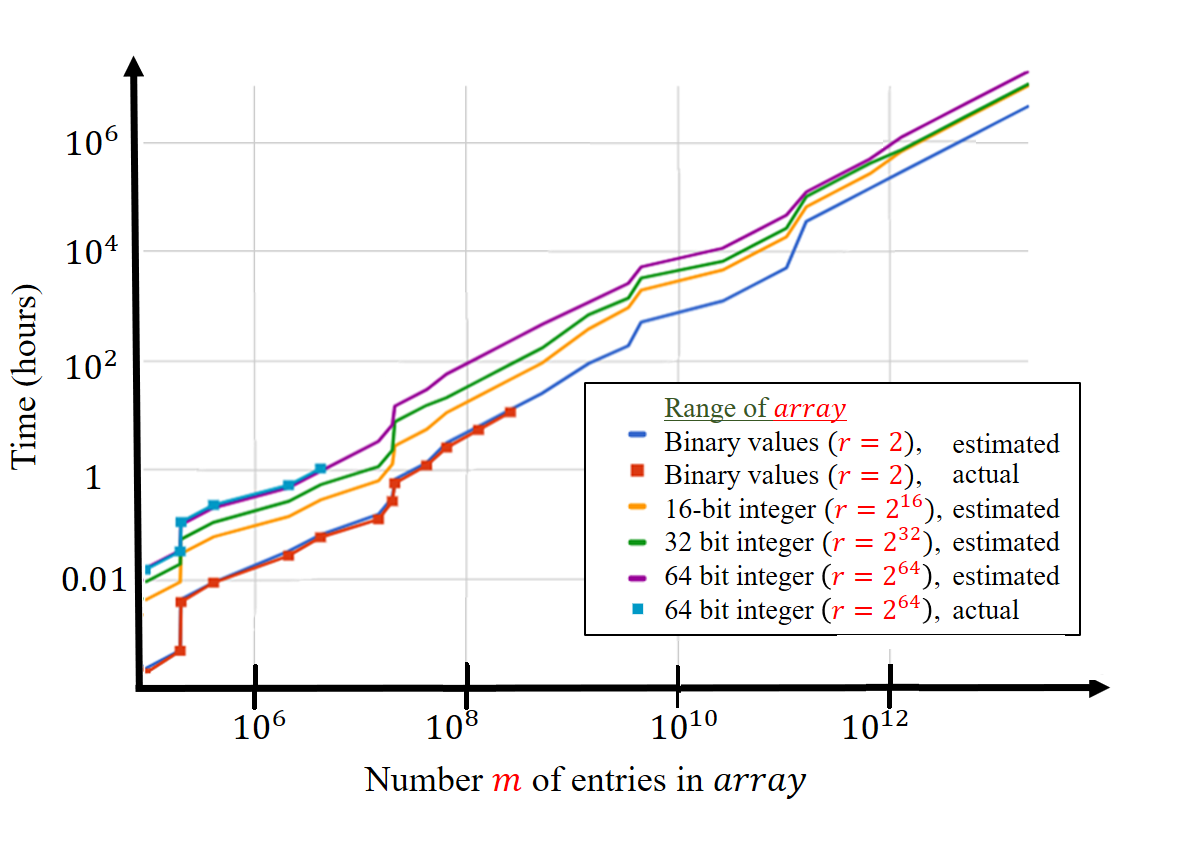}
\caption{Server's running time ($y$-axis) on a single machine on Amazon's cloud, for different database $\cost$ size ($x$-axis) of Secure Search (Algorithm~\ref{alg:FirstPositive-ClientSide1}) over encrypted database. Each colored curve represents a different range $r$ of integers. The red dots represent actual experiments, and the other curves are based on Formula~\ref{for}, which seems remarkably accurate.
}
\label{fig:spirit_graphs}
\end{figure}

%
\section{Practical Search Time Estimation Formula}\label{sec:formula}
%

When we move from theory to implementation, there are few additional factors to take into account. In this section we explain
them and give a generic but simple formula for the estimated running time 
of our algorithm, based on this more involved analysis. In the next section and in Fig.~\ref{fig:spirit_graphs} we show that indeed the formula quite accurately predicts the experimental results, at least for the configurations of our system that we checked.

We assume that the search is for a lookup value in
$\cost\in\br{0,\cdots,r-1}^m$. Our formula for the overall running time
is then
\begin{equation}
\label{for} \boxed{T = n (1+\ceil{\log_2 n}) \cdot ADD +  \lceil \log_2 r\rceil \cdot  MUL + 2n \cdot IsPOSITIVE },
\end{equation}
where
\begin{itemize}
\item $n = \cfrac{m}{CORES \cdot SIMD}$
\item $CORES$ is the number of computation machines (practically, number of
    core processors that work in parallel)
\item $SIMD$ (Single Instruction Multiple Data) is the amount of integers
    that are packed in a single ciphertext. This $SIMD$ factor is a
    function of the ring size $p = O(\log^2 n)$ and $L = \log(d+\log s) =
    O(\log_2\log_2 n)$ for $d,s$ upper bounds on the degree and size of
    evaluated polynomials; in HELib, this parameter can be read by calling
    EncryptedArray::size(), see~\cite{HaleviShoup2014helib}.
\item $ADD$,
    $MUL$, and $IsPOSITIVE$ are the times for computing a single addition,
    multiplication, and the $\isPositive$ command (see
    Algorithm~\ref{alg:ispositive}), respectively, in the contexts of
    parameters $p$ and $L$.
\end{itemize}

For example, in our system (see next section) on input parameters range $r=1$ 
and $m=255,844,736$ array entries, we have $SIMD = 122$ and $CORES=64$ resulting in $n=32,767$ packed ciphertexts;
ring size $p=17$;
and measured timings of $ADD=0.123$ms, $MUL=62.398$ms, $IsPOSITIVE=695.690$ms.
See Fig.~\ref{fig:spirit_graphs} for graph of $T$ on various $m,r$ parameter.

\paragraph{Intuition behind Formula~\eqref{for}} Theorem~\ref{lem:FirstPositive-ClientSide} states a running time that is near-logarithmic in $m$.
This is due to Definition~\ref{def:secure-search} 
that allows us to evaluate unbounded
number of polynomials in parallel, so using $m$ machines we can search all
the entries of the array in parallel. When using only $CORES \ll m$ machines,
each machine suffers a running time slowdown by a factor of $m/CORES$. 
SIMD enables parallel
computation that admits additional parallelization factor of $SIMD$, and the overall
slowdown for a machine is thus $n$ as defined above. The final time $T$ is
then the number of multiplication, additions and calls to $\isPositive$ that
are used by the SPiRiT algorithm.

\paragraph{Power of SPiRiT. }In a first look it seems that SPiRiT uses
$O(m^2)$ additions and mulitiplications, since this is the size of its
matrices $S,P,R,T$. However, since these are binary matrices, their
multiplication by a vector can be implemented by using only sum of subsets of
items with no multiplications. Moreover, these matrices are usually sparse.
Hence, the running time $T$ is near-linear in $n$ for each machine, as indeed
occurs in practice; see Fig.~\ref{fig:spirit_graphs}.

%
\section{System and Experimental Results}\label{sec:xp}
%

In this section we describe the secure search system that we implemented using the algorithms in this paper.
To our knowledge, this is the first implementation of such a search system.
Our system can roughly search 100 Gigabytes of data per day using a cloud of 1000 machines on Amazon EC2 cloud. As our experiments show, the running time reduces near-linearly with the number of machines in a rate of 100 Mega bytes per day per machine. We expect that more advanced machines would significantly improve our running times, including existing machines on Amazon EC2 cloud that use e.g. more expensive GPUs.

The system is fully open source, and all our experiments are reproducible. We hope to extend and improve the system in future
papers together with both the theoretical and practical community.

\subsection{The System}

\paragraph{System Overview. }
We implemented the algorithms in this paper into a system that maintains an encrypted database that is stored on Amazon's AWS cloud. The system gets from the client an encrypted lookup value $\ell$ to search for, and a column name $\cost$ in a database table of length $m$.
\removed{Encrypting the column name is optional, incurring a running time growth by a factor logarithmic in the number of database columns.}
The encryption is computed on the client's side using a secret key that is unknown to the server.
The client can send the request through a web-browser, that can be run e.g. from a smart-phone or a laptop.
The system then runs our secure search algorithm on the cloud, and returns a $k$-search coreset for $(\cost, \ell)$; see
Definition~\ref{def:sketch-firstpositive}.
The web browser then decrypts this coreset on the client's machine and uses it to compute the smallest index $i^*$ in $\cost$ that contains $\ell$, where $i^*=0$ if $\ell$ is not in $\cost$.
As expected by the analysis, the decoding and decryption running time on the
client side is negligible and practically all the search is done on the
server's side (cloud). Database updates can be maintained between search
calls, and support multiple users that share the same security key.

\paragraph{Hardware. }Our system is generic but in this section we evaluate it on Amazon's AWS cloud.
We use one of the standard suggested grids of EC2 {\bf x1.32xlarge} servers, each with 64 2.4 GHz Intel Xeon E5-2676 v3
(Haswell) cores and 1,952 GigaByte of RAM. Such cores are common in standard laptops.

\paragraph{Open Software and Security. }The algorithms were implemented in $C++$. HELib library~\cite{HELib} was used for the FHE
commands, including its usage of SIMD (Single Instruction Multiple Data)
technique. The source of our system is open under the GNU v3 license and can
be found in~\cite{}. Our system and all the experiments below use a security
key of 80 bits of security.
This setting can be easily changed by the client.

\subsection{Experimental Results}\label{xp:dis}


In this sub-section we describe the experiments we run on our system and explain the results.

\paragraph{Data. }We run the system on a lookup value $\ell$ in a $\cost$ of integers over two ranges: $r=2$ and
$r=2^{64}$. In the first case the vector was all zeroes except for a random
index, and in the second case we use random integers. As expected from the
analysis, the running time of our algorithm depends on the range $r$ and
length $m$ of $\cost$, but not on the actual entries or the desired lookup
value. This is since the server scans all the encrypted data anyway, as it
cannot tell when the lookup value was found.

\paragraph{The Experiment. }We run our search algorithm as described in
Section~\ref{sec:FirstPositive-Client} for vectors (database table columns)
of different length, ranging from $m=10$ to $m=100,000,000=10^8$ entries, and
for $r\in\br{2,2^{64}}$ as explained above.

\paragraph{Results.} Our experimental results for each machine on the cloud are summarized in the square points in 
Fig~\ref{fig:spirit_graphs} and also in Table~\ref{table_spirit_running_time}.
The client's decoding time was negligible in all
the experiments, so the server's time equals to the overall running time. For
example, from the graph we can see that the on a single machine we can search in a single day an array of more than
250,000,000 binary entries, and an estimate of 30,000,000 entries with values between $0$ and $2^{64}-1$.
We also added additional 6 curves that are based on the search time
Formula~\eqref{for}. As can be seen, the formula quite accurately predicts
the actual running time.

%


\paragraph{Scalability} The running time on each machine was almost identical (including the non-smooth steps; see below).
Finally, since the parallel computations on the machines are almost
independent (``embarrassingly parallel"~\cite{wilkinson1999parallel}), the
running time decreases linearly when we add more machines (cores) to the
cloud, as expected.


\paragraph{Comparison to theoretical analysis. }In Theorem~\ref{lem:FirstPositive-ClientSide} we proved that that the overall running of the search
algorithm is only poly-logarithmic in $m$ and $r$ using a parallel call to $m$ polynomials, i.e., a sufficiently large cloud.
Based on the linear scalability property above, for each single machine we thus expect to get running time that is near-linear
in $m$ and $r$. This is indeed the case as explained in the previous paragraph.

\paragraph{Why the curves are not smooth? }Each of the curves in Fig.~\ref{fig:spirit_graphs} has 4--5 non-continuous increasing steps. These
are not artifacts or noise. They occur every time where the set $\Primes$ of
primes, which depend on the length $m$, changes; see
Algorithm~\ref{alg:FirstPositive-ClientSide1}. As can be seen in the
analysis, this set changes the ring size $p$ that are used by the SPiRiT
sketch (see Algorithm~\ref{alg:FirstPositive-Server}), which in turn
increases the depth of the polynomial that realizes $\isPositive$ (see
Algorithm~\ref{alg:ispositive}), which finally increases the server running
time.

These steps are predicted and explained by the search time
Formula~\eqref{for} via the ceiling operator over the logs that make the time
formula piecewise linear.

\begin{table} [h!]
\begin{center}
\begin{tabular}{ |c|c|c| }
\hline
$m$ (vector size) &  binary (sec) & 64-bit (sec)\\
\hline
90,048        &   0.72               &    56.44\\
\hline
192,960       &   1.86               &    121.69\\
\hline
196,416       &   14.47              &    415.98\\
\hline
399,168       &   32.67              &    851.97\\
\hline
2,048,256     &   101.18             &    1,960.36\\
\hline
4,112,640     &   219.67             &    3,967.93\\
\hline
14,520,576    &   468.56             &    8,340.32\\
\hline
19,641,600    &   999.51             &    16,758.59\\
\hline
20,699,264    &   2,129.65           &    51,946.41\\
\hline
41,408,640    &   4,506.36           &    96,065.39\\
\hline
63,955,328    &   9,559.24           &\\
\hline
127,918,464   &   20,079.59          &\\
\hline
255,844,736   &   42,193.79          &\\
\hline
255,844,736   &   42,193.79          &\\
\hline
\end{tabular}
\end{center}
\caption{Server's running time of Secure Search (Algorithm~\ref{alg:FirstPositive-ClientSide1}) as measured on a single machine on Amazon's cloud for different database $\cost$ size (left column) over encrypted database. The middle column
shows the running times in seconds for a binary input vector and the right column shows the running tmes in seconds for a vector of 64-bit integers ($r=2^{64}$).}
\label{table_spirit_running_time}
\end{table}

%% file: sections/generic-search-short-version.tex
%
\section{Extension to Generic Search}\label{sec:generic-search}
%

We discuss here extensions of our results.
First, our results extend to address searching for approximate rather than exact match, or more generally, for a generic definition of what constitutes a match to the lookup value. 
Second, while we assumed for simplicity that the input is given in binary representation, our results extends to other representations. 
Finally, the output of our algorithm can include the value $\cost(i^*)$ on top of the index $i^*$.
These extensions are discussed in the following.

\paragraph{Generic search.}
The goal of generic search is to search for a match to a lookup value $\ell$ in a (not necessarily sorted) encrypted array, where what constitute a match is defined by an algorithm $\isMatch(a,b)$ that returns $1$ if-and-only-if $a,b$ are a match and $0$ otherwise.
\begin{definition}[Generic Search\label{def:generic-search}.]
For $\isMatch(\cdot,\cdot)$ returning binary values, on the input $\cost\in \br{0,\cdots,r-1}^m$ and lookup value $\ell\in\br{0,\cdots,r-1}$, the goal of \emph{generic search} is to output
\[
i^* = \min\sett{i\in[m]}{\isMatch(\cost(i),\ell)=1}
\]
using a non-interactive algorithm.
\end{definition}

\begin{theorem}[Secure Generic Search]\label{thm:generic-search}
    Suppose $\isMatch()$ is realized by a polynomial of degree $d$ and log-size $s$. Then there is an algorithm for the generic search problem that is computed in client time $O(\log^3 m)$ and a non-interactive parallel call to $O(\log^3 m)$ polynomials, each of log-size $O(\log m + \log s)$ and degree $O(d \log^4 m)$.
\end{theorem}
\begin{proof}
  The algorithm for generic search is our Algorithm~\ref{alg:FirstPositive-ClientSide1}, where the only difference is in the implementation of $\tobinary$ where we replace the calls to $\isEq$ with calls to $\isMatch$.
\end{proof}

This may be useful, for example, in the context of bio-informatics DNA alignment problems, or, pattern-matching problems in general, with $\isMatch()$ returning $1$ on all entries $\cost(i)$ whose \emph{edit-distance} from the lookup value is smaller than a user defined threshold. Similarly, the distance metric for defining the desired matches may be the \emph{Hamming distance} in the context of error-correcting-codes; the \emph{Euclidean distance} for problems in computational geometry; the \emph{Root-Means-Square (RMS)} for some machine-learning problems, or any other measure of prediction-error to be minimized; etc.
Our results show that if $\isMatch$ can be computed efficient by the server, then the resulting algorithm is efficient. The former is known to hold for metrics of interest, including the hamming and edit distance~\cite{hamming15,Cheon2015}.

\paragraph{Generic input representation.}
To handle input given in non binary representation, the only change needed is in the algorithm $\tobinary$, where we replace $\isEq$ with a testing equality in the given representation. When this equality test is realized by a polynomial of degree $d$ and size $s$, then the resulting $\tobinary$ algorithm is realized by $m$ polynomials of degree $d$ and size $s$ (executed in parallel). In particular, for native representation in ring $\ZZ_p$ we can use for equality-testing the polynomial $\isZero_p(a-b)$ of degree $d = p-1$ and size $s=2$.

\paragraph{Returning value together with index.}
To return the value $\cost(i^*)$ on top of the index $i^*$ we utilize known techniques (\eg, from~\cite{DorozSunar2014pir,DorozSunar2016pir4search}) for returning $\cost(i)$ given $i$ with a linear degree polynomial. To keep the algorithm non-interactive, the search coreset items returned from Algorithm~\ref{alg:oneWitness} include both the binary representation $b$ of an index $i$ and the value $\cost(i)$; in Algorithm~\ref{alg:FirstPositive-ClientSide1} the client then outputs both $i^*$ and $\cost(i^*)$ (for $i^*$ as specified there).

%% file: sections/prelim2.tex
%
\section{Preliminaries}\label{sec:prelim}
%

In this section we specify some standard definitions and simple algorithmic tools to be used as part of our algorithms.

\subsection{Definitions and Notations}

%

\begin{definition}[Polynomial\label{polydef}]
A polynomial is an expression that can be built from constants and symbols called variables by means of addition, multiplication and exponentiation to a non-negative integer power.
That is, a polynomial can either be zero or can be written as the sum of a finite number of non-zero terms, called monomials. The degree of a monomial $x_1^{a_1} \cdot x_2^{a_2} \cdot \ldots \cdot x_t^{a_t}$ in variables $x_1,\ldots, x_t$ is the sum of powers $a_1+a_2+\ldots+a_t$.

The \emph{degree $d=\deg(p)$ of a polynomial $p$} is the maximal degree over all its monomials. The \emph{size $s=\sizep(p)$} of this polynomial is the number of its monomials. The \emph{log-size} of $p$ is $\log_2s$, and its complexity is the sum $\comp(p)=O( d + \log s)$ of its degree and log-size.
\end{definition}

For example, the polynomial $p(x)=x^2 + x + 12$ has degree $\deg(p)=2$, size $\size(p)=3$ and log-size $\log_23$. Its complexity is then $\comp(p)=O(2+\log_23)=O(1)$.

While the values $x_1,x_2$ may be based on other variables, such as an input to an algorithm, the polynomials $p_1,p_2,\cdots$ in this paper are assumed to be fixed, i.e., universal constants as in the following definition.
\begin{definition}[Universal constant\label{global}] A data structure or a parameter is a universal (global) constant if it is independent of the input data. It may still depend on the size of the data.
\end{definition}
This is somehow related to the class P/poly of poly-time algorithms that get poly-length advice, where advice depends only on input size~\cite{lutz1993circuit}. We do not need to compute a universal constant as part of our main algorithm. Instead, we can compute it only once and, e.g., upload it to a public web-page.
We can then ignore its construction time in our main algorithm, by passing the constant as an additional input. A universal constant may depend on the input size such as $r$ and $m$, but not on the dynamic input values of $\ell$ or $\cost$. If the exact array length or range is unknown, or is a secret, we can use $r$ and $m$ only as upper bounds. This is also why we measure the pre-processing time it takes to compute the required universal constant, but do not add it to the overall running time of the algorithm itself.

As common in cryptography, the arithmetic operations in this paper are applied on a finite set of integers in $\br{0,\cdots,\rr-1}$, where the modulo operator is used to keep the outcome of each operation in this set. Such a set is formally called the $\ZZ_\rr$ ring. We denote the modulo of two integers $a\geq 0$ and $b\geq 1$ by $(a \mod b)=a - (b \cdot  \lfloor a/b\rfloor)$. More generally, for a vector $v\in\REAL^m$, we define $(v \mod b)$ to be a vector in $\REAL^m$ whose $i$th entry is $v(i) \mod b$ for every $i\in[m]$.

\begin{definition}[Ring $\ZZ_\rr$\label{ring}]
The \emph{ring} $\ZZ_\rr$ is the set $\br{0,\cdots,\rr-1}$ equipped with multiplication $(\cdot)$ and addition $(+)$ operations modulo $\rr$, i.e., $ a\cdot b = \big((a\cdot b) \mod \rr\big)$ and $a+b= \big((a+b) \mod \rr\big)$ for every $a,b\in\ZZ_\rr$.
\end{definition}

%
\subsection{Algorithmic Tools}
%

In this section we present simple algorithms for the following tasks: (1) Algorithm~\ref{alg:EQ-binary} tests equality of two bit-strings; (2) Algorithm~\ref{alg:ToBinaryExact}, given a vector of bit-strings $\cost$ and a bit-string $\ell$ to look for, returns the indicator vector accepting $1$ on all entries $i$ where $\cost(i)=\ell$ and $0$ otherwise; (3) Algorithm~\ref{alg:ispositive} given a vector of integers returns a binary vector with each entry $x$ mapped to $0$ if $x$ is a multiple of the underlying ring modulus $p$, and $1$ otherwise.

\subsection{Comparison of Binary Vectors}\label{sec:iszero-binary}

The comparison of two numbers given by their binary representation (which can be considered as binary vectors over $\ZZ_p$, not necessarily $p=2$) is particularly useful and simple; see Algorithm~\ref{alg:EQ-binary} and the discussion below.
\begin{algorithm}
    \label{alg:EQ-binary}
    \caption{$\isEq_t(a,b)$}
{
\begin{tabbing}
    \textbf{Parameters:} \= An integer $t\ge 1$ \\
    \textbf{Input:} \> Binary representations (vector) $a=(a_1,\ldots ,a_t)$ and $(b_1,\ldots,b_t)$ in $\zo^t$ \\ 
    \textbf{Output:} \> Return $1$ if-and-only-if $a=b$, and $0$ otherwise.
\end{tabbing}
}
\vspace{-0.3cm}
    $y \gets \prod_{i=1}^t \left(1 - (a_i-b_i)^2\right)$\\
    \Return $y$
\end{algorithm}

For two bits $a,b\in\zo$, their squared difference $(a-b)^2$ is $0$ if-and-only-if they are equal, and it is $1$ otherwise. Hence, $\isEq_1(a,b) = 1-(a-b)^2$ is a degree $2$ polynomial for the equality-test, evaluating to $1$ if-and-only-if $a=b$ and $0$ otherwise. For bit-strings $a=(a_1,\ldots ,a_t)$ and $(b_1,\ldots,b_t)$ in $\zo^t$, the equality-test simply computes the $AND$ of all bit-wise equality test with the degree $2t$ polynomial: $\isEq_t(a,b) = \prod_{i=1}^t \left(1-(a_i-b_i)^2\right)$. The correctness of the equality-test holds over the ring $\ZZ_p$ for every $p\ge 2$. Hence, the polynomial that corresponds to Algorithm~\ref{alg:EQ-binary} can be over any such ring.
%

\begin{lemma}\label{lem:eq-binary}
Let $t\ge 1$ be an integer. Let $a=(a_1,\ldots ,a_t)$ and $(b_1,\ldots,b_t)$ be a pair of binary vectors in $\zo^t$. Let $y$ be the output of a call to Algorithms $\isEq_t(a,b)$; see Algorithm~\ref{alg:EQ-binary}. Then $y=1$ if-and-only-if $a=b$, and $y=0$ otherwise. Moreover, Algorithm~\ref{alg:EQ-binary} can be realized by a polynomial of degree and log-size $2t$.
\end{lemma}

\subsection{Reduction to Binary Vector}\label{sec:tobinary-exact} 

Algorithm $\tobinary$ reduces a given pair of vector $\cost\in\set{0,\ldots,r-1}^m$ and lookup value $\ell\in\set{0,\ldots,r-1}$ to the indicator vector $\iz\in\zo^m$ indicating for each $i\in[m]$ whether $\cost(i)$ and $\ell$ are an exact match. That is, $\iz(i)=1$ if-and-only-if $\cost(i)=\ell$.

For simplicity of the presentation we assume here that the input values $\cost(i)$ and $\ell$ are given in binary representation and that we seek exact match $\cost(i)=\ell$. Nevertheless, $\tobinary$ algorithm easily extends to handle other input representations, as well as approximate search; see Section~\ref{sec:generic-search}. 

%

\begin{algorithm}
    \label{alg:ToBinaryExact}
    \caption{$\tobinary_{m,r}(\cost,\ell)$}
{\begin{tabbing}
\textbf{Parameters:} \=Two integers $m,r\geq 1$. \\
%
\textbf{Input:} \> $\cost\in\br{0,\cdots,r-1}^m$, a lookup value $\ell\in\br{0,\cdots,r-1}$ \\
\> where values are given in binary representation using $t=O(\log r)$ bits.\\
%
\textbf{Output:} \> $\iz\in \zo^m$ such that $\iz(i) = 1$ if $\cost(i))=\ell$ and $0$ otherwise.
\end{tabbing}\vspace{-0.3cm}}
\For {$\each$ $i\in[m]$\label{l18}}{
    $\iz(i) \gets \isEq_t(\cost(i),\ell)$  \label{l19}\tcc{see Algorithm~\ref{alg:EQ-binary}.}
}
\Return $\iz$\label{l20}
\end{algorithm}

\paragraph{Algorithm Overview.}
In Lines~\ref{l18}--\ref{l19}, for each $i\in[m]$, we assign to $\iz(i)$ the outcome of the matching test $\cost(i)=\ell$. This is done via a call to Algorithm~\ref{alg:EQ-binary}. 
The output vector $\iz$ is returned in Line~\ref{l20}.

\begin{lemma}
\label{lem:tobinary}
Let $m,r\geq1$ be integers, $\cost\in\br{0,\cdots,r-1}^m$ and $\ell\in\br{0,\cdots,r-1}$.
Let $\iz$ be the output of a call to Algorithm $\tobinary_{m,r}(\cost,\ell)$; see Algorithm~\ref{alg:ToBinaryExact}.
Then $\iz$ is a binary vector of length $m$ satisfying that for all $i\in[m]$,
\[
\iz(i)=1 \text{ if-and-only-if }\cost(i)=\ell,
\]
and $0$ otherwise. Moreover, Algorithm~\ref{alg:ToBinaryExact} can be realized by $m$ independent polynomials, each of $O(\log r)$ degree and log-size.
\end{lemma}

\subsection{Positive vs. Zero Test} \label{ispossec}

To run SPiRiT over low-degree polynomials, we implement $\isPositive_{p,t}$ from Section~\ref{sec:spirit} via a polynomial over a ring of size $p<\infty$. In particular, Algorithm~\ref{alg:ispositive} tests whether a given integer $x$ in $\set{0,\ldots,p-1}$ is zero $(=0$) or positive ($>0$). More generally, if $x$ is a vector of $t$ entries, this test is applied for each entry of $x$.
\anote{adi's change. erased text:As explained in the proof,}
Its correctness follows immediately from the Euler's Theorem \cite{hardy75}.

\begin{algorithm}
    \label{alg:ispositive}
    \caption{$\isPositive_{p,t}(x)$}
{\begin{tabbing}
\textbf{Parameters:} \=A prime integer $p\geq 2$, and a positive integer $t\ge 1$.\\
\textbf{Input:} \>A vector $x\in\br{0,\cdots, p-1}^{t}$.\\
\textbf{Output:} \>A vector $y\in\zo^{t}$ with entries $y(i)=0$ if-and-only-if $x(i)=0$ ($1$ otherwise).
\end{tabbing}}
\vspace{-0.3cm}
        Set $y(i) \gets (x(i))^{p-1} \mod p$, for every $i\in[t]$\\
    \Return $y=(y(1),\cdots,y(t))$
\end{algorithm}
%

\begin{lemma}\label{lem:ispositive}
Let $p\geq2$ be a prime, $t\ge 1$ an integer, and $x\in\br{0,\cdots,p-1}^t$. Let $y$ be the output of a call to $\isPositive_{p,t}(x)$; see Algorithm~\ref{alg:ispositive}. Then $y\in\zo^t$ satisfies that for every $i\in[t]$, $y(i)=0$ if-and-only-if $x(i)=0$. Moreover, Algorithm~\ref{alg:ispositive} can be realized by a parallel call to $t$ independent polynomials, each of degree $p-1$ and size $1$, resulting in server time $O(p)$.
\end{lemma}

%% file: sections/proofs-search.tex
%
%


%
\section{SPiRiT over $\REAL$: Proof of Lemmas \ref{lem:roots} and Theorem~\ref{lem:firstpositive-reals}}
\label{sec:analysis-first-prositive-reals}
%

In this section we prove Theorem \ref{lem:firstpositive-reals} showing that $\sfirstindex_{m,\infty,r}(x)$ returns the first positive index of $x$. A part of this proof is the analysis of our implementation of the roots matrix (see Claim~\ref{claim:analysis-root-reals}) which provides a proof for Lemma~\ref{lem:roots}.
%

\begin{proof}[of Theorem~\ref{lem:firstpositive-reals}]
The proof follows immediately from the claims below showing that analysis of each individual component of $\spirit$.
Specifically, by the claims below, for every $x\in\REAL^m$, the following holds:
\begin{itemize}
  \item
  $w=Tx$ is the tree-representation of $x$.

  \item
  $w' = \isPositive_{\REAL,2m-1}(w)$ satisfies that $w'(k)=1$ if-and-only-if there is a leaf with non-zero key among the leaves rooted at node (indexed by) $k$.

  \item
  $v = Rw'\in \REAL^m$ satisfies that: (i) $v\in\set{0,\ldots,d}^m$ for $d = O(\log m)$ the depth of the tree representation of $x$, and (ii) $v(j)=0$ if-and-only-if $\sum_{i=0}^j x(i)=0$. Moreover, (iii) the rows of $R$ are $d$-sparse (\ie, each row has at most $d$ non-zero entries).

  \item
  $u = \isPositive_{\REAL,m}(v) \in\zo^m$ represents a step-function $(0,\ldots,0,1,\ldots,1)$ whose first $1$ value is at entry $i^*$, for $i^*$ the first positive entry of $x$.

  \item
  $Pu\in\zo^m$ has a single non-zero entry at $i^*$.

  \item
  $St \in\zo^{\log m}$ is the binary representation of the index $i^*$ for the unique non-zero entry in $t$.
\end{itemize}
We conclude therefore that output $b=\spirit\cdot x$ is the binary representation of the first positive entry in $x$, as required.
\end{proof}

\begin{claim}[Tree Matrix $T\in\zo^{(2m-1)\times m}$]
For every $x\in\REAL^m$, $Tx \in \REAL^{2m-1}$ is the tree-representation of $x$ (see Definition~\ref{tree}).
\end{claim}
\begin{proof}
    The Tree Matrix $T\in\zo^{(2m-1)\times m}$: Fix $x\in\REAL^m$, and a row $k$ in $T$. The $k$th row corresponds to a node $u$ in the tree representation and is the indicator vector for the leaves rooted at $u$, thus the inner product of this row with vector $x$ gives the sum of keys over the leaves rooted at $u$. Namely, $w = Tx \in \REAL^{2m-1}$ is the tree-representation of $x$.
\end{proof}

\begin{claim}[$\isPositive$ operator on $Tx$]
For every $x\in\RR^m$ and $w=Tx\in\RR^{2m-1}$, 
$w' = \isPositive_{\REAL,2m-1}(w) \in\zo^{2m-1}$ satisfies that $w'(k)=1$ if-and-only-if there is a leaf with non-zero key among the leaves rooted at node (indexed by) $k$.
\end{claim}
\begin{proof}
    The $\isPositive_{\REAL,2m-1}$ operator applied on $Tx$: Fix $x\in\RR^m$ and $w=Tx\in\RR^{2m-1}$ the tree representation of $x$. Let $w' = \isPositive_{\REAL,2m-1}(w)\in\zo^{2m-1}$. Then, 
    $w'(k)=1$ if-and-only-if $w(k)\neq 0$, where by the definition of tree representation, the latter holds if-and-only-if among the leaves rooted at (node indexed by) $k$ there is a leaf with non-zero key.
\end{proof}

\begin{claim}[Roots Sketch $R\in\zo^{m\times (2m-1)}$ with our suggested implementation in Lemma~\ref{lem:roots}]\label{claim:analysis-root-reals}
For every $x\in\REAL^m$ and $w'=iTx\in\zo^{2m-1}$, $v = Rw'\in \REAL^m$ satisfies that:
(i) $v\in\set{0,\ldots,d}^m$ for $d = \log m$ the depth of the tree representation of $x$; (ii) the rows of $R$ are $d$-sparse (\ie, each row has at most $d$ non-zero entries); and (iii) $v(j)=0$ if-and-only-if $\sum_{i=1}^j x(i)=0$.
\end{claim}
\begin{proof}
    The Roots Sketch $R\in\zo^{m\times (2m-1)}$: Let $v = RiTx$ and fix an entry $j\in[m]$.
    To prove (i)-(ii), observe that
    \[
    \card{\lop(j+1)}\le d
    \]
    because a leaf in a full binary with $m=2^d$ leaves has at most $d+1$ ancestors, and at most $d$ of them are left-siblings, as the root is excluded.
    Now, since $iTx(k)\in\zo$ we conclude that $v(j) = \sum_{k\in\lop(j+1)} iTx(k)$ accepts values in $\set{0,\ldots,d}$; and since the number of $1$ entries in each row $j\in[m]$ of $R$ is $\card{\lop(j+1)}$ we conclude that $R$ is $d$-sparse.

    To prove (iii), 
    first observer that
    \[
    \sum_{k\in\lop(j+1)} Tx(k) = \sum_{i=1}^j x(i)
    \]
    because the left-siblings (corresponding to indexes) in $\lop(j+1)$ partition the leaves $x(1),\ldots,x(j)$ into disjoint sets (with a set for each left-sibling, containing all leaves $x(i)$ for which the left-sibling is an ancestor), and $Tx(k)$ is the sum of values $x(i)$ over all leaves for which that left-sibling (indexed by) $k$ is an ancestor.
    Next observe that by definition of the $\isPositive$ operator $i()$,
    \[
    \sum_{k\in\lop(j+1)} Tx(k)\neq 0\text{\ \ \ \ if-and-only-if \ \ \ \ }\sum_{k\in\lop(j+1)} iTx(k) \neq 0.
    \]
    \details{because $\sum_{k\in\lop(j+1)} Tx(k) \neq 0$ if-and-only-if $Tx(k)\neq 0$ for some $k\in\lop(j+1)$, which in turn holds if-and-only-if $iTx(k)=1$ and $\sum_{k\in\lop(j+1)} iTx(k) \neq 0$.\footnote{Looking ahead, this holds also under finite ring $p>\card{\lop(j+1)}=O(\log m)$.}}
    Finally, since $v(j) = \sum_{k\in\lop(j+1)} iTx(k)$, we conclude that $v(j) \neq 0$ if-and-only-if $\sum_{i=1}^j x(i)\neq 0$.
\end{proof}

\begin{claim}[$\isPositive$ operator on $v = RiTx$]\label{claim:iRiT-reals}
For every $x\in\REAL^m$ and $v=RiTx$, 
$u = \isPositive_{\REAL,m}(v) \in\zo^m$ 
represents a step-function $(0,\ldots,0,1,\ldots,1)$ whose first $1$ value is at entry $i^*$, for $i^*$ the first positive entry of $x$.
\end{claim}
\begin{proof}
    The $\isPositive_{\REAL,m}$ operator applied on $v = RiTx$: Fix every $v\in\REAL^m$. The vector $u = \isPositive_{\REAL,m}(v) \in\zo^m$ satisfies that $u(j)=1$ if-and-only-if $v(j)\neq 0$. In particular, for $v\in\set{0,\ldots,d}$ s.t. $v(j)=0$ if-and-only-if $\sum_{k=0}^j x(k)=0$, $u(i)=0$ for all $i<i^*$, and $u(i)=1$ for all $i\ge i^*$.
\end{proof}

\begin{claim}[Pairwise Sketch $P\in\zo^{m\times m}$]
For every $u\in\zo^m$ representing a step function $(0,\ldots,0,1,\ldots,1)$ with $1$ value at entry $i^*$, $Pu\in\zo^m$ has a single non-zero entry at $i^*$.
\end{claim}
\begin{proof}
        Pairwise Sketch $P\in\zo^{m\times m}$: Observe that $Pu(i) = u(i)-u(i-1)$ is the discrete derivative of $u$, which is equal to zero on every entry $i$ where $u(i)=u(i-1)$ and non-zero otherwise. For $t$ a vector representing a step-function with transition from $0$ to $1$ values at entry $u(i^*)=1$, we get the value $Pu(i^*)=1-0=1$ on $i^*$ and values $Pu(i)=u(i)-u(i-1)=0$ for all $i\neq i^*$.
\end{proof}

\begin{claim}[The Sketch $S\in\zo^{\log m\times m}$]
For every $1$-sparse $t \in\zo^m$, 
then $St \in\zo^{\log m}$ is the binary representation of the index $i^*$ for the unique non-zero entry in $t$.
\end{claim}
\begin{proof}
    The Sketch $S\in\zo^{\log m\times m}$: Proof follows immediately from taking $S$ to be a $(1,m)$-sketch. It is easy to see that the specified matrix is indeed a $(1,m)$-sketch: Note that the first column of $S$ is the binary representations of $0$, the next column is the binary representation of $1$, and so forth, to the last column specifying the binary representation of $m$. For $t$ with a single non-zero entry at entry $i^*$, the product $St$ is simply the $i^*$-th column of $S$, which is in turn simply the binary representation of the index $i^*$.
\end{proof}